\documentclass[10pt,journal]{IEEEtran}
\makeatletter
\def\ifundefined{\@ifundefined}
\makeatother

\newcommand{\subparagraph}{}

\usepackage{epsfig}
\usepackage{amsmath}
\usepackage{amssymb} 
\usepackage{graphicx}
\usepackage{amsthm}
\usepackage{multirow}
\usepackage{booktabs}
\usepackage{threeparttable}
\usepackage{ifthen}
\usepackage[compact]{titlesec}
\usepackage[letterpaper,textwidth=7.153in,textheight=9.6in,top=.54in]{geometry}

\newboolean{longver}
\setboolean{longver}{true}
\newcommand{\bitm}{\begin{itemize}}
\newcommand{\eitm}{\end{itemize}}
\newcommand{\be}{\begin{equation}}
\newcommand{\ee}{\end{equation}}
\newcommand{\bea}{\begin{eqnarray}}
\newcommand{\eea}{\end{eqnarray}}
\newcommand\ba[1]{\left[ \begin{array}{#1}}
\def\ea{\end{array}\right]}
\def\nn{\nonumber\\}
\newcommand{\bfi}{\begin{figure}}
\newcommand{\efi}{\end{figure}}

\newcommand{\Kendall}{Kendall tau~}

\newcommand{\fn}{}

\newtheorem{thm}{Theorem}  %   number by chapter
\newtheorem{lem}{Lemma}        %   same counter as thm

\newtheorem{pro}{Proposition}

\newtheorem{eg}{Example}

\newtheorem{defn}{Definition}       %   Numbered independently from thm
\DeclareMathOperator{\length}{length}
\DeclareMathOperator{\Cen}{Cen}
\DeclareMathOperator{\dist}{dist} 

\begin{document}

%\title{Linear Programming Upper Bounds on Permutation Code Sizes w.r.t. Kendall-Tau Distance}
\title{Linear Programming Upper Bounds on Permutation Code Sizes From Coherent Configurations Related to the Kendall Tau Distance Metric}

\author{\authorblockN{Fabian Lim$^1$ and Manabu Hagiwara$^2$}\\
\authorblockA{$^1$Research Laboratory of Electronics, Massachusetts Institute of Technology, 
Cambridge, MA 02139\\ 
$^2$National Institute of
Advanced Industrial Science and Technology\\
Central 2, 1-1-1 Umezono, Tsukuba City,
Ibaraki, 305-8568, JAPAN\\
Email: flim@mit.edu, hagiwara.hagiwara@aist.go.jp \vspace*{-10pt}
\thanks{F. Lim recieved support from NSF Grant ECCS-1128226.}
}}

\maketitle
%\vspace*{-20pt}
\begin{abstract} 
Recent interest on permutation rank modulation shows the \Kendall metric as an important distance metric.
This note documents our first efforts to obtain upper bounds on optimal code sizes (for said metric) ala Delsarte's approach. For the Hamming metric, Delsarte's seminal work on powerful linear programming (LP) bounds have been extended to permutation codes, via association scheme theory. For the \Kendall metric, the same extension needs the more general theory of coherent configurations, whereby the optimal code size problem can be formulated as an extremely huge semidefinite programming (SDP) problem. 
Inspired by recent algebraic techniques for solving SDP's, we consider the dual problem, and propose an LP to search over a subset of dual feasible solutions.
%propose simplifications that may weaken the bounds, but allow us to solve smaller LP's instead.
We obtain modest improvement over a recent Singleton bound due to Barg and Mazumdar. 
We regard this work as a starting point, towards fully exploiting the power of Delsarte's method,
which are known to give some of the best bounds in the context of binary codes.
\end{abstract}


\thispagestyle{empty}
\pagestyle{empty} 

\section{Introduction}

A permutation code is designed to only allow certain pairwise distances between any two codewords. These codes have been studied in various contexts, \emph{e.g.}, group codes~\cite{Slepian1968}, signal modulation~\cite{Slepian1965,Mittelholzer1996}, vector quantization~\cite{Goyal2001}, rank modulation~\cite{Jiang2009,Barg2009}, cost-constrained transpositions~\cite{Farnoud2010}, etc. This work is motivated by a recent study on a fundamental coding problem. In~\cite{Barg2009} they looked at optimal code sizes with respect to the \Kendall distance metric. This metric is important to rank modulation and its applications, \emph{e.g.}, flash memories. 

For binary codes, Delsarte's optimization-based methods~\cite{Delserte} give some of the best known bounds~\cite{Navon2007}.
For permutation codes, we observe during initial experiments (for very small lengths) that Delsarte-like methods outperform Hamming (sphere packing) bounds~\cite{Barg2009,Margolius2001}. 
Our interest is to investigate, if this improvement carries over for larger codes. 
%We aim to complement the asymptotic analysis in~\cite{Barg2009}, with (finite code length) Delserte-type linear programming upper bounds~\cite{Delserte}. 
Tarnanen extended Delsarte's work over to permutation codes~\cite{Tarnanen}, however only for the Hamming metric (and other metrics with similar symmetries).
%extensions of Delserte's work has been made by  to permutation codes~\cite{Tarnanen,Bogaerts}, though only for the Hamming metric (and other metrics with similar symmetries). 
The novelty here is to consider the \Kendall metric, and as pointed out in~\cite{Barg2009}, lacks required symmetry to straight-forwardly apply Tarnanen's techniques.

Delsarte's (and Tarnanen) techniques are based on association schemes, from which linear programming (LP) formulations (of the optimal code size problem) are obtained. 
For the \Kendall metric, one needs to consider the more general theory of coherent configurations (CC), which instead deliver semidefinite programming (SDP) formulations. 
The matrices in these SDP's turn out to be of unwieldy size, but recent work~\cite{ISDP,Schrijver2005,SDP_Dual} suggest possible approaches.
%this issue may be possibly tackled.
One may exploit the algebraic structure of the CC's, to only work with block-diagonalized (and possibly much smaller) versions of these matrices.  
To our knowledge, such recent techniques are new in the area of permutation codes.
However, the solution is not straight-forward. 
As code lengths increase, the CC's (related to the \Kendall metric) become huge quickly, motivating the techniques presented in this preliminary report.

%
% whereby
%a very large scaling of the number adjacency matrices in the related coherent configurations prevent said task from being straight-forward.
%In this preliminary study, we introduce simple ideas to obtain upper bounds at reduced complexity.

%Recent interest is due to key applications in coding for flash memories. For binary codes, Delserte pioneered combinatorial techniques leading to linear programming problems, which deliver upper bounds on optimal code sizes~\cite{Delserte}. For permutation
%codes, Tarnanen has noticed a certain generalization~\cite{Tarnanen,Bogaerts}. Unfortunately Tarnanen's generalization is limited to distance metrics (on permutation groups) that satisfy certain symmetry properties, e.g. Hamming metric. Our interest here is the recently studied Kendall-Tau metric, in which it was pointed out that the related adjacency matrix is not distance-transitive, hence the symmetry properties are lacking~\cite{Arya}. 

While we believe to be presently unable to fully exploit the power of SDP bounds, we show some initial success. We consider the dual problem (also a SDP), and use an LP to search over a subset of feasible solutions.
We obtained modest improvement over a recently published Singleton bound in~\cite{Barg2009}.
The reduced complexity allows us to compute up to permutation codes of length $11$ (where the matrices were previously of order $11$ factorial).
Certain bottlenecks, if solved, could allow computation for longer codes.
%. The former method is more complex, but is 
%We present the best linear programming bound from a certain standpoint of coherent configurations.
%Using some new ideas, we propose a simplification of the semidefinite programming problems into linear programming problems. 
%We present the best linear programming bound from the standpoint of coherent configurations.
%The reduction in complexity allows us to compute up to the symmetric group on 20 objects, that has $20!$ permutations. 
%Our comparisons with SDP approaches (computed for small symmetric groups on $3$ to $5$ objects) suggest much room for improvement. 
%However the techniques presented here are currently the next best alternative for larger symmetric groups, where SDP approaches require huge computation power.  
%As it stands, our proposed LP bounds perform poorer than known Hamming bounds~\cite{Barg2009}.
%Still, this note motivates the interesting and challenging task, to pursue further improvement (of our rather simple techniques) and better exploit the sophistication of SDP-based approaches.
As it stands, our proposed LP bounds perform poorer than known Hamming bounds~\cite{Barg2009}, and it remains to see how far sophisticated SDP-based approaches can ultimately bring us. This note aims to motivate new research to resolve this open question.

%however this note motivates the interesting and challenging task, to pursue further improvement (of our rather simple techniques) and better exploit the sophistication of SDP-based approaches.

%Still these new techniques allow initial estimates on code sizes can be obtained, which are otherwise difficult to obtain for larger symmetric groups. 
%We hope to provide insights which may lead to future improvements of these results.
%
%up to size 20.
%
%We compare our 
%
%For optimal code sizes w.r.t. Kendall-Tau distance, Barg \& Mazumdar gave a nice asymptotic analysis~\cite{Arya}. This work is motivated by the lack of fixed block analysis, in this regard a generalization of Delserte-type bounds would be desirable. We begin by looking at semidefinite programming (SDP) problems, that deliver upper bounds on optimal code sizes. Motivated by the large number of optimization variables in our SDP problems, we propose simplifications that deliver linear programming (LP) problems. The proposed simplifications are inspired (yet somewhat different) from recent work on symmetry exploitation in SDP's, see e.g.~\cite{Schrijver2005,ISDP,SDP_Dual,SDPCode}. The problems encountered here are different from previous works, hence new techniques need to be invented.

\section{Background}

\subsection{Optimal Code Size Problem and Two Metrics}

\newcommand{\Sym}{\mathcal{S}}
\newcommand{\Om}{\Omega}
\newcommand{\dmin}{\delta_{\scriptsize \mbox{\lowercase{min}}}}
\newcommand{\g}{g}
\newcommand{\h}{h}
\newcommand{\Iv}{^{-1}}
\newcommand{\define}{\stackrel{\Delta}{=}}
\newcommand{\wo}{w_0}
\newcommand{\C}{\mathcal{C}}

\newcommand{\D}{\dmin}
\renewcommand{\S}{\mathcal{V}}
\newcommand{\nD}{(n,\D)}
Let $\Sym_n$ denote the \emph{symmetric group} on a set $\{1,2, \dots, n \}$
and $\dist( , )$ be a distance metric on $\Sym_n$. 
A subset $\S$ of $\Sym_n$ is an $\nD$ \textbf{permutation code} (with respect to dist(,)), if for any $g,h\in \S$ such that $g\neq h$,
%any two \emph{different} permutations $g$ and $h$ in $\S$, \emph{i.e.}, $g\neq h$,
we have $\dist(g,h) \geq \D$. 
%In this work we are interested in optimal code sizes.
%The setup is as follows. 

\newcommand{\Asize}{\mu}

%\Asize(&n,\D) \define
\begin{defn}[Optimal code size problem] \label{def:OptCodeSize}
Let $\dist(,)$ be a distance metric on the symmetric group $\Sym_n$. Let $\dmin\geq1$.
The following problem is the \textbf{optimal code size problem}.
\begin{align}
	 & \max_{\S \subseteq \Sym_n} \# \S \label{eqn:opt1} \\
	\mbox{s.t.} \dist(\g,\h) \geq \D &\mbox{ for all } \g,\h \in \S \mbox{ where } \g \neq \h, \nonumber
\end{align}
and $\# \S$ denotes the cardinality of the set $\S$.
Denote $\Asize(n,\D)$ to be the maximal cardinality achieved by $\nD$ codes, \emph{i.e.}, $\Asize(n,\D)$ equals the optimal value of the above problem.
\end{defn}

\newcommand{\ind}{i}
The \emph{image} of $\ind$ by $g$ is denoted $g(i)$. The \emph{inverse} of $g$ is denoted $g\Iv$. 
The \emph{product} of permutations $g$ and $h$ is denoted $gh$, whereby $(gh)(\ind) = g(h(\ind))$.
Most literature (\textit{e.g.}, Tarnanen~\cite{Tarnanen}) consider the \textbf{Hamming metric}
%The metric $\dist(,)$ is commonly taken to be the \textbf{Hamming metric}
\bea
	\dist(\g,\h) &\define& \#\{ 1 \le x \le n : (\g\Iv \h)(\ind) \neq \ind\}, \label{eqn:Hamming}
%         \mbox{ for all } \g, \h \in \Sym_n, \label{eqn:Hamming}
\eea
\textit{i.e.}, the Hamming distance $\dist(\g,\h)$ equals the number of \emph{moved} points of $\g\Iv h$.
For the \emph{direct product} group $\Sym_n\times\Sym_n$, define its action 
%of $\Sym_n\times\Sym_n$ 
on $\Sym_n$, as $(\alpha, \beta) \cdot g \define \alpha g \beta \Iv$, where $( \alpha, \beta) \in \Sym_n\times\Sym_n$ and $g \in \Sym_n$.
%Define the action of any $( \alpha, \beta) \in \Sym_n\times\Sym_n$ on any $g $ in $ \Sym_n$, as $(\alpha, \beta) \cdot g \define \alpha g \beta \Iv$.
For any subgroup $\mathcal{G}$ of $\Sym_n\times\Sym_n$, 
a metric $\dist(,)$ on $\Sym_n$ is \textbf{$\mathcal{G}$-invariant} if for any $g,h\in \Sym_n$, we have 
$\dist(g,h) = \dist((\alpha,\beta)\cdot\g,(\alpha,\beta)\cdot\h) $ for all $(\alpha,\beta)\in \mathcal{G}$.
%$g, h, \alpha \in \Sym_n$,
%we have $\dist(\g,\h) = \dist(\alpha\g,\alpha\h) = \dist(\g\alpha,\h\alpha)$.
The Hamming metric (\ref{eqn:Hamming}) can be verified to be $(\Sym_n\times\Sym_n)$-invariant.

\newcommand{\Cyc}{\Psi_n}
\newcommand{\G}{\mathcal{G}}

%\begin{rem} \label{rem:bi_inv}
%In general, a metric $\dist(,)$ will be bi-invariant if $\dist(\g,\h)= \dist(\tilde{\g},\tilde{\h})$
%is equivalent to
% $\tilde{\g}\Iv\tilde{\h}$ is a conjugate of $\g\Iv\h$
% for any $g, h, \tilde{g}, \tilde{h} \in \Sym_n$
% \textit{i.e.},
% if $\tilde{\g}\Iv\tilde{\h} = \alpha\Iv(\g\Iv\h)\alpha $ for some $\alpha \in \Sym_n$.
%\end{rem}

%The metric of interest here is the \textbf{Kendall-Tau metric}.
%The Kendall-Tau metric is defined using the notion of the length of a permutation in $\Sym_n$,
%For any permutation $g\in \Sym_n$,
%%let $\length(g)$ denote the length of $g$, whereby $\length(g) = r$ equals the minimum number of \emph{adjacent transpositions} $\alpha_i \in \Sym_n$ that satisfy $g = \alpha_1 \alpha_2 \dots \alpha_r$.
%
The length of a permutation $g$, denoted $\length(g)$, equals the minimum integer $r$ satisfying $g = \alpha_1 \alpha_2 \dots \alpha_r$ whereby $\alpha_i$ are \emph{adjacent transpositions} in $\Sym_n$. For rank modulation~\cite{Jiang2009,Barg2009} we consider the \textbf{\Kendall metric}, given as
%on a symmetric group $\Sym_n$
\bea
\dist(\g,\h) \define  \length(g\Iv h). \label{eqn:length}
\eea
There exists a unique element $\wo$, termed the \textbf{longest element}, that satisfies $\length(\wo)= n(n-1)/2$.
Then $\wo$ is an involution, \emph{i.e.}, $\wo\Iv=\wo$, and $\dist(\g,\h) = \dist(\g\wo,\h\wo)$, 
%because $\wo$ is known as an \emph{anti-automorphism} between \textit{left (and right) weak-Bruhat orders}, 
see~\cite{Humphreys}, p. 119.
Denote a subgroup $\{e,\wo\}$ of $\Sym_n$ by $\Cyc$, where $e$ is the identity element of $\Sym_n$.
In general, the \Kendall metric is $(\Sym_n\times \Cyc)$-invariant.

A permutation $g$ written as $g=(123)$ means $g(1)=2$, $g(2)=3$ and $g(3)=1$. Note $(12),(23),(13)$ are transpositions, in particular the first two are adjacent transpositions.

\newcommand{\size}[1]{\# #1}
\begin{eg} \label{eg:1}
Consider $\Sym_3$ with elements $e,(12),(23),(123),\linebreak[4](132),(13)$, and the Hamming metric. The minimum distance between any two non-equal permutations is $2$. 
%Taking $\D$ in the form previously described, F
For $\dmin=1$ and $2$ 
%So for $\D = \{1,2,3\}$ and $\D = \{2,3\}$, 
we have $\Asize(n,\D) = \size{\Sym_3}$. 
%For $\D = \{3\}$, 
For $\dmin=3$ 
the code $\S$ with the optimal size satisfies $\S = \{e,(123),(132)\}$. 
Check 
%The solution of the optimal code size problem for $\D = \{3\}$ w.r.t the Hamming metric, is given by $\S = \{e,(123),(132)\}$, where 
$\dist(e,(123))= \dist(e,(132)) =3$, and $\dist((123),(132))=\dist(e,(123)\Iv(132))=\dist(e,(123)) = 3$.

%For the Kendall-Tau metric, the minimum distance between any non-equal perms. is $1$. 
The minimum possible non-zero \Kendall pairwise distance is $1$.
%For $\D=\{1,2,3\}$, 
For $\dmin = 1$, 
we have $\Asize(n,\D) = \size{\Sym_3}$ as before. 
For $\dmin = 2$ the optimal code satisfies $\S = \{e,(123),(132)\}$. 
%The solution of the optimal code size problem for $\D=\{2,3\}$ w.r.t the Kendall-tau metric, is again given by $\S = \{e,(123),(132)\}$, 
Check $\dist(e,(123)) = \length((123)) = 2$, where $(123)=(12)(23)$.
% and both $(12)$ and $(23)$ are adjacent permutations. 
%The solution of the same problem 
%For $\D = \{3\}$ 
For $\dmin = 3$ 
the optimal code satisfies $\S = \{e,(13)\}$, where $(13)$ is the longest element $\wo$ in $\Sym_3$ and $\length((13))=3$ (here $(13)=(12)(23)(12)$).
\end{eg}

%The optimal code size problem (\ref{eqn:opt1}) is NP-Hard. 
%In the sequel we will visit a (convex) relaxation, whereby combinatorial symmetries will be exploited to further reduce computations.

\newcommand{\p}{p}
\newcommand{\tup}[1]{(#1)}
\newcommand{\Orb}{\Delta}
\newcommand{\brak}[1]{\left\langle #1\right\rangle}
\newcommand{\tta}{\theta}
\newcommand{\1}{\mathbf{1}}
\newcommand{\tp}{*}
\renewcommand{\d}{d}

\newcommand{\OmSq}{{\Om^2}}
%\section{Background}

\subsection{Coherent Configurations (CC)}\label{section:CoherentConfigurationsAndAssociationSchemes}

%We now describe combinatorial objects possessing the symmetries that allow said complexity reduction. 

%We now describe objects used to formulate relaxations of (\ref{eqn:opt1}).
%Let $\G$ be a group which acts on a \emph{finite} set $\Om$. 
%Define the product set $\OmSq$ as $\OmSq \define \Om \times \Om$. Define an \textbf{induced action} of $\G$ on $\OmSq$ as $\g\cdot(x,y) \define \tup{g(x),g(y)}$, where
%$\g \in \G$ and $\tup{x,y} \in \OmSq$. An orbit of the induced action of $\G$ on $\OmSq$ is termed an \textbf{orbital}.
%%We call an orbit according to the induced action of $\G$ on $\Om \times \Om$ an \textbf{orbital}.
%The orbitals $\Orb_1,\Orb_2,\cdots,\Orb_\d$ of the action of $\G$ on $\OmSq$
% partition $\OmSq = \cup_{i=1}^\d \Orb_i$. 
% If the action of $\G$ on $\Om$ is \emph{transitive}, we use the convention $\Orb_1=\{(x,x): x \in \Om\}$.
% For each orbital $\Orb_i$, we correspond an \textbf{adjacency matrix} $A_i$ as follows. The matrix $A_i$ is a 0-1 matrix, whose rows/columns are indexed by the set $\Om$, and we have $(A_i)_{x,y} = 1$ if and only if $\tup{x,y} \in \Orb_i$. Let $A_i^T$ denote the transposed matrix of $A_i$.
 
\renewcommand{\OmSq}{\Sym_n\times\Sym_n}
We now describe objects used to formulate relaxations of (\ref{eqn:opt1}).
%Let $\G$ be a group which acts on a \emph{finite} set $\Om$. 
%Define the product set $\OmSq$ as $\OmSq \define \Om \times \Om$. 
For a subgroup $\G$ of $\OmSq$, 
define an \textbf{induced action} of $\G$ on $\OmSq$, as $\g\cdot(x,y) \define \tup{g(x),g(y)}$ where
$\g \in \G$ and $x,y\in\Sym_n$. An orbit of this induced action is termed an \textbf{orbital}.
%We call an orbit according to the induced action of $\G$ on $\Om \times \Om$ an \textbf{orbital}.
These orbitals $\Orb_1,\Orb_2,\cdots,\Orb_\d$ of the induced action
 partition $\{(x,y): x,y\in\Sym_n\}= \cup_{i=1}^\d \Orb_i$. 
 If the action of $\G$ on $\Sym_n$ is \emph{transitive}, we use the convention $\Orb_1=\{(x,x): x \in \Sym_n\}$.
 For each orbital $\Orb_i$, we correspond an \textbf{adjacency matrix} $A_i$ as follows. Here $A_i$ is a 0-1 matrix, whose rows/columns are indexed by $\Sym_n$, and we have $(A_i)_{x,y} = 1$ if and only if $\tup{x,y} \in \Orb_i$. Let $A_i^T$ denote the transposed matrix of $A_i$.
 
\renewcommand{\Om}{\Sym_n}
\newcommand{\tpr}{T}
\begin{thm}[c.f.~\cite{BI}, p. 52] \label{thm:CC}
%Let $\G$ be a group which acts on a finite set $\Om$. 
%Let $\G$ be a subgroup of $\Sym_n\times\Sym_n$.
Let $\G$ be a group which acts on $\Om$ transitively.
%Assume the action of $\G$ on $\Om$ is transitive.
%The adjacency matrices $A_1, A_2, \cdots, A_d$ corresp. to the $\d$ orbitals $\Orb_i$ 
For the induced action of $\G$ on $\OmSq$, the adjacency matrices $A_i$
%$A_1, A_2, \cdots, A_d$
 corresponding to the $\d$ orbitals $\Orb_i$, satisfy 
\bitm
	\item[i)] $A_1 $ equals the identity matrix.
	\item[ii)] the sum $\sum_{i=1}^\d A_i $ equals the all ones matrix.
	\item[iii)] for any $A_i$, there exists some $A_{i'}$ that satisfies $A_i^\tpr = A_{i'}$.
	\item[iv)] for any $i, j \in \{1,2,\cdots,\d\}$, there exists numbers $\p_{ij}^k$ that satisfy $A_i A_j = \sum_{k=1}^\d \p_{ij}^k A_k$.
\eitm
\end{thm}

\newcommand{\perm}{\pi}
\newcommand{\Real}{\mathbb{R}}

\newcommand{\Complex}{\mathbb{C}}
\newcommand{\Alg}{\mathcal{A}}

\newcommand{\lengthCC}{\tup{\Sym_n\times \Cyc, \Sym_n}}
\newcommand{\conjCC}{\tup{\Sym_n\times \Sym_n, \Sym_n}}
\newcommand{\Bm}[1]{A_{\Cyc,#1}}
\newcommand{\Am}[1]{A_{\Sym_n,#1}}

%Associated with group $\G$ and set $\Om$, 
A \textbf{coherent configuration (CC)} denoted $(\G,\Om)$, refers to the set $\{A_1,A_2,\cdots, A_\d\}$ of corresponding adjacency matrices. 
A CC with the additional property $\p_{ij}^k = \p_{ji}^k$ is an \emph{association scheme}; in this special case, Delsarte showed how combinatorial properties can deliver linear programming (LP) bounds~\cite{Delserte}.
%In his Ph.D. thesis, Delserte considered CC's $(\G,\Om)$ endowed with an additional property $\p_{ij}^k = \p_{ji}^k$, specially called \emph{association schemes}~\cite{Delserte}. Delserte exploited their combinatorial properties to obtain linear programming (LP) bounds. These CC's $(\G,\Om)$ are related to our optimal code size problem, by invariances  of the metrics $\dist(,)$.
%By the action of $\Sym_n\times\Sym_n$ on $\Sym_n$, 
Construct two CC's related to the $\G$-invariances of the Hamming and \Kendall metrics.
For the former metric, set $\G = \Sym_n \times \Sym_n$ and call $\conjCC$ the \textbf{conjugacy CC} - the name comes from~\cite{Tarnanen}. 
For the latter metric, set $\G= \Sym_n \times \Cyc$ and term $\lengthCC$ the \textbf{length CC}. 
Let $\Real^{\Sym_n \times \Sym_n}$ denote the set of real matrices and index set $\Sym_n$. Write $\Am{i}$ and $\Bm{i}$ for adjacency matrices of conjugacy, and length CC.

\begin{eg} \label{eg:2}
%Consider $\Sym_3$. 
The matrices in $\Real^{\Sym_3\times \Sym_3}$ corresponding to the conjugacy and length CC (the indexing on $\Sym_3$ is done in the same order that appears in Eg. \ref{eg:1}), are written as follows. First $\Am{1}=\Bm{1}=I$, where $I$ is the identity matrix. Next
\[
\Am{2} = 
\left[
\begin{array}{cccccc}
0     & 1     & 1     & 0     & 0     & 1 \\
1     & 0     & 0     & 1     & 1     & 0 \\
1     & 0     & 0     & 1     & 1     & 0 \\
0     & 1     & 1     & 0     & 0     & 1 \\
0     & 1     & 1     & 0     & 0     & 1 \\
1     & 0     & 0     & 1     & 1     & 0 \\
\end{array}
\right],~
\Bm{2} = 
\left[
\begin{array}{cccccc}
0     & 1     & 1     & 0     & 0     & 0 \\
1     & 0     & 0     & 1     & 0     & 0 \\
1     & 0     & 0     & 0     & 1     & 0 \\
0     & 1     & 0     & 0     & 0     & 1 \\
0     & 0     & 1     & 0     & 0     & 1 \\
0     & 0     & 0     & 1     & 1     & 0 \\
\end{array}%
\right].
\]
By Theorem \ref{thm:CC}, $\Am{3}=J - I -\Am{2}$, 
%Next $\Am{3}=\Bm{3}=J - I -\Am{2}$, 
here $J$ has all ones. Finally, it so happens that we get $\Bm{3} = \Am{3}$ and $\Bm{4}  = \Am{2} - \Bm{2}$.
%Also $\Am{1}=\Bm{1}=I$, where $I$ is the identity matrix, and $\Am{3}=\Bm{3}$ and $\Am{2} = \Bm{2}+\Bm{4}$. 
%Let $J$ denote the all-ones matrix. From Thm. \ref{thm:CC} we deduce $\Am{3}=J - I -\Am{2}$. 
%We deduce $\Am{3}$ using prop. ii), Thm. \ref{thm:CC}.
Matrices $\Am{1}$ to $\Am{3}$ corresponding to Hamming distances $0,2,3$, and $\Bm{1}$ to $\Bm{4}$ to \Kendall distances $0,1,2,3$. 
%Matrices $\Bm{1},\Bm{2},\Bm{3},\Bm{4}$ correspond to Kendall-Tau distances $0,1,2,3$. 
\end{eg}

The focus here is on the length CC, related to the \Kendall metric. The conjugacy CC (related to the Hamming metric) is actually an association scheme, and is treated in~\cite{Tarnanen}; 
the recollection is because of connections exploited later.
%the reason we reprise it here is due to some intimate connections with the length CC, to be discussed in the sequel.

\section{Semidefinite Programming (SDP) Bounds} \label{sect:SDP}

\newcommand{\Csize}{k}
\newcommand{\A}{A}
\newcommand{\x}{b}
\newcommand{\Tr}{ \mbox{\upshape Tr}}
\renewcommand{\v}{\nu}
\renewcommand{\g}{x}
\renewcommand{\h}{y}
\newcommand{\Asym}{\tilde{A}}
\newcommand{\dsym}{{\tilde{d}}}
\newcommand{\Orbsym}{\tilde{\Orb}}
\newcommand{\del}{\delta}

A symmetric matrix $M $ in $ \Real^{\Sym_n\times\Sym_n}$ is \emph{positive semidefinite}, if all its eigenvalues are \emph{non-negative}.
We now use CC's to formulate the relaxation of the optimal code size problem.
%We now visit a (convex) relaxation of the optimal code size problem. %, whereby combinatorial symmetries of CC's will be exploited to further reduce computations.
By iii), Theorem \ref{thm:CC}, a set $\{\Asym_1,\Asym_2,\cdots, \Asym_\dsym\}$ of \textbf{symmetrized adjacency matrices} are obtained, whereby $\dsym\leq \d$. %~For $\tup{\G,\Sym_n}$, 
If $A_i$ is not symmetric, then find $A_{i'}$ such that $A_i^\tpr = A_{i'}$, and set $\tilde{A}_j = A_i + A_{i'}$. Similarly the \textbf{symmetrized orbitals} $\Orbsym_i$ are obtained by setting $\Orbsym_j = \Orb_i \cup \Orb_{i'}$ if $\Asym_{j} = A_{i} + A_{i'}$.
Note both $(g,h)$ and $(h,g)$ belong to the same $\Orbsym_j$,
%lie in the same \emph{symmetrized} orbital, 
and $\dist(g,h)=\dist(h,g)$. % for any metric $\dist(,)$. 
Thus by $\G$-invariance of $\dist(,)$ set $ \delta_j = \dist(g,h) $ for any $(g,h)\in \Orbsym_j$, since $\dist(g,h)=\dist(g',h')$ %for any two $(g,h)$ and $(g',h')$ in the same $\Orbsym_j$. 
for any $(g,h), (g',h')\in\Orbsym_j$. 
The values $\delta_j$ are called \textbf{orbit-distances}
%We term the set $\{\del_2,\del_3,\cdots, \del_\dsym\}$ the \textbf{distance set} 
(with respect to a $\G$-invariant metric $\dist(,)$). If $\G$ acts transitively on $\Sym_n$, then by convention $\Orbsym_1=\{(g,g): g \in \Sym_n\}$, thus $\del_j\geq1$ for all $j\geq 2$.
%then the distance set must contain only \emph{positive} values by our convention $\Orbsym_1=\{(g,g): g \in \Sym_n\}$.
The properties of the CC's can simplify the following optimizations. %problems.
%In the sequel, we exploit properties of the CC to simplify the following optimization problems.

\begin{defn}[Primal SDP, $\tup{\G,\Sym_n}$ and $\D$] \label{def:SDP}
Let $\G$ be a group which acts on $\Sym_n$ transitively and $\dist(,)$ a $\G$-invariant distance metric on $\Sym_n$. 
Let $\del_j$ be the orbit-distances w.r.t. $(\G,\Sym_n)$ and $\dist(,)$.
%Let $\del_j$ be positive integers for $2 \le j \le \dsym$
%and
% $\D$ a subset of $\{ \del_2, \dots \del_\dsym \}$.
Define the \textbf{semidefinite programming (SDP)} problem correp. to $\tup{\G,\Sym_n}$ and some $\D\geq 1$, as %\vspace*{-10pt}
\begin{align} %\setminus \{0\}
	  &\max_{M \in \Real^{\Sym_n \times \Sym_n} } \Tr{(J M)} \label{eqn:SDP}  \\
	\mbox{s.t. }& M~\mbox{is positive semidefinite, and } \Tr(M)=1, \nn
	&\begin{array}{cll}
%		\Tr(M)  				&= 1  \\
		\Tr(\Asym_j M)  &\geq 0, &\mbox{ for } 2 \le j \le \dsym,\nn
		\Tr(\Asym_j M)  &= 0, &\mbox{ for } 2 \le j \le \dsym \mbox{ with } \del_j < \D,
	\end{array}
\end{align}
where $\Asym_j$ is a corresponding symmetrized adjacency matrix,
$J$ is the all-one matrix, and
$\Tr$ is the trace function.
\end{defn}

\newcommand{\syma}{\dagger}
\newcommand{\symb}{\ddag}
\newcommand{\symc}{\diamond}

\begin{table}
	\centering
	\begin{threeparttable}[t]
	
	\caption{[Initial Experiments] SDP Bounds for $3\leq n \leq 5$}
% Table generated by Excel2LaTeX from sheet 'isit3'
\renewcommand{\arraystretch}{.6}
\renewcommand{\c}{@{\hspace{1.2ex}}c@{\hspace{1.2ex}}}
\begin{tabular}{\c\c\c\c\c\c\c\c\c\c\c\c\c}
\toprule
\multicolumn{5}{c}{$n =3$}            & \multicolumn{5}{c}{$n =5$}            & \\
\midrule
$\dmin$ & $b_1^*$ & SB    & HB    & \textit{Search} & $\dmin$ & $b_1^*$ & SB    & HB    & \textit{Search} \\
1     & \textbf{6} & 6     & 6     & \textit{\textbf{6}} & 1     & \textbf{120} & 120   & 120   & \textit{\textbf{120}} \\
2     & \textbf{3} & 6     & 6     & \textit{\textbf{3}} & 2     & \textbf{60} & 120   & 120   & - \\
3     & \textbf{2} & 2     & 2     & \textit{\textbf{2}} & 3     & \textbf{22} & 120   & 24    & - \\
\multicolumn{5}{c}{$n =4$}            & 4     & \textbf{14} & 120   & 24    & - \\
$\dmin$ & $b_1^*$ & SB    & HB    & \textit{Search} & 5     & \textbf{7} & 24    & 8     & - \\
1     & \textbf{24} & 24    & 24    & \textit{\textbf{24}} & 6     & \textbf{5} & 24    & 8     & \textit{3} \\
2     & \textbf{12} & 24    & 24    & \textit{\textbf{12}} & 7     & \textbf{3} & 24    & 4     & \textit{2} \\
3     & \textbf{5} & 24    & 6     & \textit{\textbf{5}} & 8     & \textbf{2} & 6     & 4     & \textit{\textbf{2}} \\
4     & \textbf{3} & 6     & 6     & \textit{\textbf{3}} & 9     & \textbf{2} & 6     & 2     & \textit{\textbf{2}} \\
5     & \textbf{2} & 6     & 2     & \textit{\textbf{2}} & 10    & \textbf{2} & 2     & 2     & \textit{\textbf{2}} \\
6     & \textbf{2} & 2     & 2     & \textit{\textbf{2}} &       &       &       &       &  \\
\bottomrule
\end{tabular}%

	\label{tab:SDP}
	\begin{tablenotes}
	\item [${}^\syma$] Singleton bound (SB), published in~\cite{Barg2009}, equation (5).
	%\item [${}^\symb$] Hamming bound (HB), see~\cite{Barg2009} equation (6). %obtained using ball sizes calculated from known numbers of permutations with $k$ inversions,~see~\cite{Margolius2001}.
	\item [${}^\symb$] Hamming bound (HB) from ball-size estimates, see~\cite{Margolius2001,Barg2009}. 
	\item []  {\em Note: Above table created by taking numerical floor.}
	\end{tablenotes}
	\end{threeparttable}
	\vspace*{-15pt}
\end{table}

%\vspace*{-5pt}
\begin{pro}\label{pro:SDPrelax}
Let $\G$ be a group which acts on $\Sym_n$ transitively
and
$\dist(,)$ a $\G$-invariant distance metric on $\Sym_n$.
%$\dist(,)$ is invariant under the action of $\G$ on $\Sym_n$.
%Let $\Orbsym_1, \Orbsym_2, \cdots, \Orbsym_\dsym$ be the symmetrized orbitals
%of the induced action of $\G$ on $\Sym_n \times \Sym_n$.
%By our convention $\Orbsym_1= \{(g,g): g \in \Sym_n\}$.
%Assume that the action of $\G$ on $\Sym_n$ is transitive, and 
% %%% its ok to leave this out?
%Set $\del_j = \dist(g, h)$ where $(g, h) \in \Orbsym_j$. 
%Let $\Orbsym_1= \{(g,g): g \in \Sym_n\}$ and $\dist(\Orbsym_1)=0$.
%Let $\D$ be a subset of positive integers in the set $ \{ \del_2, \del_3, \cdots, \del_\dsym\}$.
%$ \{ \del( \Orbsym_2 ), \del( \Orbsym_3 ), \cdots, \del( \Orbsym_\dsym) \}$. %where we also set $\del( \Orbsym_j ) = \del_j$ for $\del_j$ appearing in Definition \ref{def:SDP}.
%Let $\D$ be a subset of the distance set 
%corresp. to $(\G, \Sym_n)$ and $\dist(,)$.
Let $\dmin \geq1$. 
%Let $\Asize(n, \D)$ be the optimal objective value of (\ref{eqn:opt1}) for $\dist(,)$ and $\D$.
%Then, the optimal objective value of (\ref{eqn:SDP}) %, in which $\del_j = \del( \Orbsym_j ) $,
% upper bounds $\Asize(n,\D)$.
Then, the optimal objective value of (\ref{eqn:SDP})
 upper bounds the optimal objective value %$\Asize(n,\D)$
 of (\ref{eqn:opt1}) for $\dist(,)$ and $\D$.
\end{pro}

\ifthenelse{\boolean{longver}}
{The SDP (\ref{eqn:SDP}) is a relaxation of the optimal code size problem (\ref{eqn:opt1}), see appendix for proof.}
{The SDP (\ref{eqn:SDP}) is a relaxation of the optimal code size problem (\ref{eqn:opt1}), see~\cite{FH} for the proof.} 
The optimal value of the SDP (\ref{eqn:SDP}) is at most $\size{\Sym_n}$, as for any feasible $M$, we have $\Tr(JM)\leq \Tr(J)=\size{\Sym_n}$. 
Software like SeDuMi~\cite{Sturm1999} can solve SDP's.
%For reasonably small $\size{\Sym_n}$, the SDP can be solved (as is) by optimization packages like SeDuMi~\cite{Sturm1999}. 
% but one has to compute with matrices of order $\size{\Sym_n}$.

\begin{eg} \label{eg:3}
Consider $\G=\Sym_3 \times \Psi_3$, whereby the \Kendall metric is $\G$-invariant. 
Let $\Orbsym_1$ to $\Orbsym_4$ correspond to $\Bm{1}$ to $\Bm{4}$ (all symmetric). 
%Let $\Orbsym_1,\Orbsym_2,\Orbsym_3,\Orbsym_4$ correspond to $\Bm_1,\Bm_2,\Bm_3,\Bm_4$ (all symmetric). 
%We have $\del(\Orbsym_1)=0$, $\del(\Orbsym_2)=1$, $\del(\Orbsym_3)=2$, and $\del(\Orbsym_4)=3$.
%Let $M^*$ maximize the SDP (\ref{eqn:SDP}). %Let $J$ denote the all-ones matrix of size $6\times 6$.
Using SeDuMi we solve for 
%$\D=\{1,2,3\}$, $\D=\{2,3\}$ and $\D=\{3\}$ (or 
$\dmin =1$, 2 and 3, and get the optimal solutions %$M^*= \frac{1}{6} \cdot J$,  
\[
\frac{1}{6} \cdot J,~~~
\frac{1}{6} \cdot (\Bm{1} +\Bm{3}),~~~
\frac{1}{6} \cdot (\Bm{1}+\Bm{4}).
\] 
%which corresponds to $\Tr(J M^*)$ given as $6,3$ and $2$. 
which correspond to optimal objective values $6,3$ and $2$. 
%We can verify that in all cases the matrix $M^*$ is positive semidefinite, and the constraint $\Tr(M^*)=1$ is satisfied. Also for the matrices $\Bm_2,\Bm_3,\Bm_4$, the constraint $\Tr(\Bm_j M^*) = 0$ is satisfied wherever $\del_j \notin\D$, whereby $\del_j = \del(\Orbsym_j)$. In this special case, the upper bounds $\Asize(n,\D) \leq \Tr(J M^*)$ delivered by the SDP (\ref{eqn:SDP}), are tight (recall Example \ref{eg:1}).
\end{eg}

%\vspace*{-5pt}
 
We need to work with the dual problem to (\ref{eqn:SDP}). 

\begin{defn}[Dual problem, $\tup{\G,\Sym_n}$ and $\D$] \label{def:SDP_Dual}
%Let $\G$ acts on $\Sym_n$.
%Assume the action of $\G$ on $\Sym_n$ is transitive.
%For a CC $\tup{\G,\Sym_n}$, 
Let $\G$ be a group which acts on $\Sym_n$ transitively and $\dist(,)$ a $\G$-invariant distance metric on $\Sym_n$. 
%Let $\D$ be a subset of the distance set corresp. to $(\G,\Sym_n)$ and $\dist(,)$.
Let $\del_j$ be the orbit-distances w.r.t. $(\G,\Sym_n)$ and $\dist(,)$.
Let $\Asym_j$ be a corresponding symmetrized adjacency matrix to $(\G,\Sym_n)$. 
%Let $\del_j$ be positive integers for $2 \le j \le \dsym$
%and $\D$ a subset of $\{ \del_2, \dots \del_\dsym \}$. 
Let $\dmin \geq 1$.
Define the following %problem
\begin{align}
      & \min_{(b_1, b_2, \dots, b_{\dsym}) \in \Real^{\dsym}} b_1 \label{eqn:SDP_Dual} \\
\mbox{s.t. }
    & b_j \leq 0\mbox{ for } 2 \le j \le \dsym \mbox{ with } \del_j \geq \D, \nn
        & \sum_{j=1}^\dsym b_j \Asym_j  - J \mbox{ is positive semidefinite}, \nonumber
\end{align}
to be the \textbf{dual problem} of the SDP in Definition 4.
%where $J$ is the all-ones matrix.
\end{defn}

Any feasible solution $b $ in $\Real^{\dsym}$ to the dual program (\ref{eqn:SDP_Dual}), provides an upper bound to the optimal objective value of the SDP (\ref{eqn:SDP}), see~\cite{SDP_Dual}; we have the following chain of inequalities
\bea
\Tr(J M^*) \leq b_1^* \leq b_1, \label{eqn:chain}
\eea
where $M^*$ and $b^*$ are optimal solns. of (\ref{eqn:SDP}) and (\ref{eqn:SDP_Dual}), resp.
%the SDP and the dual problem, respectively. 
%Our motivation to look at SDP bounds comes from the following observation.

Our interest in SDP bounds is motivated by initial experimentation.
Table \ref{tab:SDP} shows optimal objective values of (\ref{eqn:SDP_Dual}) obtained using SeDuMi, for (small) $n=3$ to $5$. We compare with two other bounds, i) a \emph{Singleton bound} (SB) recently published in~\cite{Barg2009}, and ii) a \emph{Hamming bound} (HB) obtained by sphere packing, see~\cite{Barg2009}. Ball-sizes for HB were obtained from \emph{exact} numbers of permutations with $k$ inversions~\cite{Margolius2001}. 
For cases shown, SDP bounds always perform the best, 
with some tightness verified by limited \emph{exhaustive searches}.
Given that optimization-based bounds are (some of) the best-known for binary codes, e.g. see discussion in~\cite{Navon2007}, 
it is not unusual to ask: \textbf{for permutation codes, are SDP bounds always better for all $n$?} %even for larger $n$? 

\newcommand{\dthe}{{\tilde{d}_{\Theta_n}}}
\begin{table}
	\centering
	\caption{Number $d$ of adjacency matrices}
		% Table generated by Excel2LaTeX from sheet 'isit2'
		\renewcommand{\arraystretch}{.6}
%\begin{tabular}{cccc|cccc}
%\toprule
%%$n$   & $\dsym$ & $\d$  & $\dthe$ & $n$   & $\dsym$ & $\d$  & $\dthe$ \\
%$n$   & Len. & Conj.  & $\dthe$ & $n$   & Len. & Conj.  & $\dthe$ \\
%\midrule
%4     & 13    & 5     & 8     & 8     & 20352 & 22    & 182 \\
%5     & 45    & 7     & 22    & 9     & 90720 & 30    & 1184 \\
%6     & 230   & 11    & 34    & 10    & 907200 & 42    & 1300 \\
%7     & 1388  & 15    & 144   & 11    & 9979200 & 56    & 12170 \\
%\bottomrule
%\end{tabular}%

% Table generated by Excel2LaTeX from sheet 'isit3'
\begin{tabular}{cccc|cccc}
\toprule
%$n$   & $\dsym$ & $\d$  & $\dthe$ & $n$   & $\dsym$ & $\d$  & $\dthe$ \\
$n$   & Len. & Conj.  & $\dthe$ & $n$   & Len. & Conj.  & $\dthe$ \\
\midrule
4     & 13    & 5     & 8     & 8     & 10558 & 22    & 171 \\
5     & 45    & 7     & 21    & 9     & 92126 & 30    & 860 \\
6     & 230   & 11    & 34    & 10    & 912908 & 42    & 1052 \\
7     & 1388  & 15    & 122   & 11    & 9998008 & 56    & 7578 \\
\bottomrule
\end{tabular}%

	\label{tab:dsym}
	\vspace*{-15pt}
\end{table}

\renewcommand{\fn}{\footnote{For the case $\Sym_6$ we tried using the techniques in~\cite{ISDP} to produce results for Table \ref{tab:SDP}, but SeDuMi was unable to compute.}}
To seek an answer we should compute for larger $n$, thus motivating the proposed method in the next section.
When ${\Sym_n}$ gets large, problems (\ref{eqn:SDP}) and (\ref{eqn:SDP_Dual}) become increasingly difficult to solve, as the matrices $\tilde{A}_j$ have order $\size{\Sym_n}$. 
Our method is inspired by recent work~\cite{ISDP,Schrijver2005,SDP_Dual}, which show that 
if $\tilde{A}_j$ come from a CC, then the
%in (\ref{eqn:SDP}) and (\ref{eqn:SDP_Dual}) are equivalent to SDP's whereby $\tilde{A}_j$ 
$\tilde{A}_j$ can be replaced (in (\ref{eqn:SDP}) and (\ref{eqn:SDP_Dual})) by \emph{block-diagonalized} versions - exact details omitted here. 
This may result in huge complexity reduction,
\emph{e.g.},~\cite{SDP_Dual} shows how SDP's related to the conjugacy CC 
%$\conjCC$  (studied in~\cite{Tarnanen}), 
reduces to simpler LP problems. % (with only $\dsym$ optimization variables).
The caveat is that number of matrix blocks (obtained from diagonalization) is at least $\d$, the number of adjacency matrices $A_i$, see~\cite{ISDP}.
Unfortunately for the length CC, this number quickly becomes large for increasing $n$, see Table \ref{tab:dsym}. 
%Thus in our case it becomes difficult\fn~to directly apply the techniques in~\cite{ISDP}, and modifications of the ideas are needed.
Thus in our case it becomes difficult to directly apply the techniques in~\cite{ISDP}, and modifications of the ideas are needed.
%%EXPLAINATION WHY \d is big for Length CC
%The values $\del_j$ in the distance set are not guaranteed to be unique (Eg. \ref{eg:2} is an exception); the number $\d$ depends on the number of orbitals $\Orb_j$ obtained from the $\G$-invariance of $\dist(,)$.
%For the Kendall-Tau metric where $\G=\Sym_n\times\Cyc$, because $\size{\Cyc}=2$ so $\d \geq \size{\Sym_n}/2$.
%%EXPLANATION ENDS
%Next, we seek alternative solutions that turn out to deliver new LP bounds for the length CC.

%The dual problem (\ref{eqn:SDP_Dual}) also requires consideration of matrices of order $\size{\Sym_n}$. 
%for the conjugacy CC $\conjCC$, Tarnanen considered linear programming (LP) problems (with only $\dsym$ optimization variables,see~\cite{Tarnanen}) that can be shown to be equivalent to (\ref{eqn:SDP_Dual}). For the length CC $\lengthCC$ things are more complicated, and (in general) the same equivalence does not hold. Our idea is to show that, if one considers a special subset of dual feasible solutions, then we can simplify the semidefinite constraint $\sum_{j=1}^\dsym b_j \Asym_j  - J$ in (\ref{eqn:SDP_Dual}), into a linear constraint. This is done by exploiting algebraic properties of the CCs, as pursued in the next section. 

\section{Length CC: Linear Programming (LP) Bounds}

\newcommand{\AlgG}{\Alg_{\G,\Sym_n}}
\newcommand{\AlgT}{\Alg_{\Sym_n \times \Cyc,\Sym_n}}
\newcommand{\AlgO}{\Alg_{\Sym_n \times \Sym_n,\Sym_n}}
\newcommand{\Is}{\tilde{\mathcal{I}}}
\newcommand{\Iss}{{\mathcal{I}}}
\renewcommand{\fn}{\footnote{The set $\AlgT$ is usually known as the adjacency algebra (over $\Real$) of the CC $(\G,\Sym_n)$, which has the properties of a matrix-$\tp$ algebra~\cite{BI}.}}

Using ``duality'' we consider the feasible solutions $b$ to (\ref{eqn:SDP_Dual}) (for some $\G$-invariant $\dist(,)$ and $\dmin\geq 1$) that furnish upper estimates $b_1$ to $\Asize(n,\D)$, see (\ref{eqn:chain}) and Proposition \ref{pro:SDPrelax}. 
While ``duality'' ideas are not new, the novelty here is to ``guess a good subset'' of feasible solutions (in the dual program) described %using only linear equations, so that we can use an LP to optimize over them.
by a manageable number of linear equations, and use an LP to optimize over them.
For a CC $(\G,\Sym_n)$,
 a feasible solution $b$ corresponds to a positive semidefinite matrix
 in the following set\fn %~$\AlgG$:
\bea
\AlgG \define \left\{ \sum_{i=1}^\dsym b_j \tilde{A}_j : b_j \in \Real, \mbox{ for all }  1\leq j\leq\dsym \right\}. \label{eqn:AlgT}
\eea
Recall that the all-ones matrix $J$ is also in $\AlgG$.

\renewcommand{\d}{{d_{\Sym_n}}}
\renewcommand{\dsym}{{d_{\Cyc}}}
To build an intuition on how such a strategy is possible,
 we first connect with the LP bound of the conjugacy CC $\conjCC$ described
 in~\cite{Tarnanen}. 
To clarify between conjugacy and length CC's,
 we respectively denote $\Am{i}$ and $\Bm{i}$ for adjacency matrices,
 and $\d$ and $\dsym$ for their numbers.

We claim that the set $\AlgO$ is a subset of $\AlgT$,
 seen by showing each $A_i$ to lie in $\AlgT$.
Observe that $\Sym_n\times\Cyc$ is a subgroup of $\Sym_n\times\Sym_n$, hence the orbitals of the length CC, lie within those of the conjugacy CC. In other words, there exists index subsets $\Iss_{\Sym_n, i}$, where $\cup_{i=1}^\d \Iss_{\Sym_n, i} = \{1,2,\cdots,\dsym\}$, such that $\Am{i} = \sum_{j\in \Iss_{\Sym_n, i}} \Bm{j}$ hold (for all $i$).
The claim $\Am{i} \in \AlgT$ follows if $\Am{i}$ is a symmetric matrix, see property i) of the following theorem from~\cite{Tarnanen}.

\newcommand{\AmT}[1]{\tilde{A}_{\Sym_n,#1}}
\begin{thm}[c.f.~\cite{Tarnanen}] \label{thm:conjCC}
Let $\tup{\Sym_n\times\Sym_n,\Sym_n}$ denote the conjugacy CC, 
%with $\d$ adjacency matrices $\Am{i}$, and $\Am{1}=I$.
where $\conjCC = \{\Am{i}: 1\leq i \leq \d\}$, and $\Am{1}=I$. 
%where  $\conjCC = \{A_1,A_2,\cdots, A_\d\}$, and $A_1 $ equals the identity matrix $I$.
Then all of the following hold for $\Am{i}$:
\bitm
	\item[i)] symmetry, {i.e.}, $\Am{i}^T = \Am{i}$ (or $\AmT{i}=\Am{i}$).
%	\item[ii)] the matrix $\tp$-algebra $\Alg_{\Sym_n\times\Sym_n,\Sym_n}$ is commutative, i.e. $\Alg_{\Sym_n\times\Sym_n,\Sym_n} =\Cen(\Alg_{\Sym_n\times\Sym_n,\Sym_n})$.
	\item[ii)] commutativity, {i.e.}, $\Am{i} \Am{j} = \Am{j} \Am{i}$ for all $i,j$.
%	\item[iii)] the matrix $\tp$-algebra $\Alg_{\Sym_n\times\Sym_n,\Sym_n}$ is spanned matrices $A_i$, and the dimension $\dimC$ of $\Cen(\Alg_{\Sym_n\times\Sym_n,\Sym_n})$ equals $\d$.
%	\item[iv)] there exists an orthonormal matrix $U$ in $\Real^{\Sym_n\times \Sym_n}$ that diagonalizes every matrix $M$ in $\Cen(\Alg_{\Sym_n\times\Sym_n,\Sym_n})$. That is, there exists $U \in \Real^{\Sym_n\times \Sym_n}$ such that for any $M\in \Cen(\Alg_{\Sym_n\times\Sym_n,\Sym_n})$ we have $U^\tp M U = \sum_{i=1}^\d a_i \cdot I_i$, whereby $a_i \in \Comp$ and $\sum_{i=1}^\d I_i = I$.
%	\item[iii)] an orthonormal matrix $U$ in $\Real^{\Sym_n\times \Sym_n}$ diagonalizes all $A_i$, {i.e.}, for each $A_i$ we have $U^T A_i U =  \sum_{j=1}^\d \p_{i,j} \cdot I_j$ for some $\p_{i,j}\in \Real$ and 0-1 diagonal matrix $I_j$. 
	\item[iii)] diagonalization by an orthonormal matrix $U$ in $\Real^{\Sym_n\times \Sym_n}$, {i.e.}, $U^T \Am{i} U =  \sum_{j=1}^\d \p_{i,j} \cdot I_j$ for some $\p_{i,j}\in \Real$ and 0-1 diagonal matrix $I_j$. 
	\bitm 
		\item[$\bullet$] $I = \Am{1} =  \sum_{j=1}^\d U I_j U^T $, therefore $\sum_{j=1}^\d I_j  = I$.
		\item[$\bullet$] $\sum_{i=1}^\d \Am{i} = J$, so $U^T J U = \sum_{j=1}^d c_j \cdot I_j$ where $c_j = \sum_{i=1}^d \p_{i,j}$. By convention $c_1=\size{\Sym_n}$ (the only non-zero eigenvalue of $J$) and $c_j=0$ for $j\geq2$.
	\eitm
%	\item[iv)] The number $d$ equals the partition number of $n$, e.g. $d$ equals $1, 2, 3, 5, 7, 11, 15, 22, 30\cdots $ for $n$ equals $1,2,\cdots, 9,\cdots$. %where $d(1)=2,d(2)=2,d(3)=3,d(4)=5,d(5)=7,\cdots$
%We also have $\sum_{i=1}^\d I_i = I$ where $I$ is the identity matrix.
\eitm
\end{thm}

\newcommand{\BmT}[1]{\tilde{A}_{\Cyc,#1}}

The numbers $\d$, tabulated in Table \ref{tab:dsym},
 equal the \emph{partition number} of $n$, see~\cite{Tarnanen}.
%Since $\AlgO\subset\AlgT$,
 Consider a matrix $\sum_{j=1}^\dsym b_j \BmT{j}$ in $\AlgT$, that for some $a\in \Real^\d$,
 can be expressed as $\sum_{i=1}^d a_i \Am{i}$. 
Theorem \ref{thm:conjCC} allows us to further express $\sum_{j=1}^\dsym b_j \BmT{j} = \sum_{j=1}^\d z_j \cdot  (U I_j U^T)$ where $z_j =  \sum_{i=1}^\d \p_{i,j} a_i $. 
Then $\sum_{j=1}^\dsym b_j \BmT{j} -J $ is positive semidefinite (see (\ref{eqn:SDP_Dual})) if the linear constraints $\sum_{i=1}^\d \p_{i,j} a_i \geq c_j$ hold for all $j$, for constants $c_j$ in iii). % Thm. \ref{thm:conjCC}.
Intuitively, Theorem \ref{thm:conjCC} is an explicit ``diagonalization'' of all matrices in the subset $\AlgO$ of $\AlgT$, and facilitates checking of positive semidef.
%. Thm. \ref{thm:conjCC} allows to check positive definiteness of any matrix in the subset $\AlgO$ using only $\d$ linear constraints.

\begin{table*}
\centering
	\begin{threeparttable}[t]
	\caption{Bounds computed for various $3 \leq n \leq 11$.}
% Table generated by Excel2LaTeX from sheet 'isit3'
\renewcommand{\arraystretch}{.6}
\begin{tabular}{rrrrrrrrrrrrrrccc}
\toprule
\multicolumn{4}{c}{$n =3$}    & \multicolumn{4}{c}{$n =7$}    & \multicolumn{3}{c}{$n =9$} & \multicolumn{3}{c}{$n=10$} & \\
\midrule
\multicolumn{1}{c}{$\dmin$} & \multicolumn{1}{c}{LP} & \multicolumn{1}{c}{SB\tnote{$\syma$}} & \multicolumn{1}{c}{HB\tnote{$\symb$}} & \multicolumn{1}{c}{$\dmin$} & \multicolumn{1}{c}{LP} & \multicolumn{1}{c}{SB} & \multicolumn{1}{c}{HB} & \multicolumn{1}{c}{$\dmin$} & \multicolumn{1}{c}{LP} & \multicolumn{1}{c}{SB} & \multicolumn{1}{c}{$\dmin$} & \multicolumn{1}{c}{LP} & \multicolumn{1}{c}{SB} \\
\multicolumn{1}{c}{1} & \multicolumn{1}{c}{\textbf{6}} & \multicolumn{1}{c}{6} & \multicolumn{1}{c}{6} & \multicolumn{1}{c}{10} & \multicolumn{1}{c}{5040} & \multicolumn{1}{c}{720} & \multicolumn{1}{c}{28} & \multicolumn{1}{c}{14} & \multicolumn{1}{c}{362880} & \multicolumn{1}{c}{40320} & \multicolumn{1}{c}{42} & \multicolumn{1}{c}{\textbf{6}} & \multicolumn{1}{c}{24} \\
\multicolumn{1}{c}{2} & \multicolumn{1}{c}{\textbf{3}} & \multicolumn{1}{c}{6} & \multicolumn{1}{c}{6} & \multicolumn{1}{c}{11} & \multicolumn{1}{c}{\textbf{630}} & \multicolumn{1}{c}{720} & \multicolumn{1}{c}{14} & \multicolumn{1}{c}{15} & \multicolumn{1}{c}{45360} & \multicolumn{1}{c}{40320} & \multicolumn{1}{c}{43} & \multicolumn{1}{c}{\textbf{3}} & \multicolumn{1}{c}{6} \\
\multicolumn{1}{c}{3} & \multicolumn{1}{c}{\textbf{2}} & \multicolumn{1}{c}{2} & \multicolumn{1}{c}{2} & \multicolumn{1}{c}{12} & \multicolumn{1}{c}{543} & \multicolumn{1}{c}{120} & \multicolumn{1}{c}{14} & \multicolumn{1}{c}{16} & \multicolumn{1}{c}{32989} & \multicolumn{1}{c}{5040} & \multicolumn{1}{c}{44} & \multicolumn{1}{c}{\textbf{3}} & \multicolumn{1}{c}{6} \\
\multicolumn{1}{c}{} & \multicolumn{1}{c}{} & \multicolumn{1}{c}{} & \multicolumn{1}{c}{} & \multicolumn{1}{c}{15} & \multicolumn{1}{c}{140} & \multicolumn{1}{c}{120} & \multicolumn{1}{c}{5} & \multicolumn{1}{c}{23} & \multicolumn{1}{c}{7560} & \multicolumn{1}{c}{720} & \multicolumn{1}{c}{45} & \multicolumn{1}{c}{\textbf{2}} & \multicolumn{1}{c}{2} \\
\multicolumn{4}{c}{$n =4$}    & \multicolumn{1}{c}{16} & \multicolumn{1}{c}{75} & \multicolumn{1}{c}{24} & \multicolumn{1}{c}{5} & \multicolumn{1}{c}{25} & \multicolumn{1}{c}{2016} & \multicolumn{1}{c}{720} &       &       &  \\
\multicolumn{1}{c}{$\dmin$} & \multicolumn{1}{c}{LP} & \multicolumn{1}{c}{SB} & \multicolumn{1}{c}{HB} & \multicolumn{1}{c}{17} & \multicolumn{1}{c}{\textbf{14}} & \multicolumn{1}{c}{24} & \multicolumn{1}{c}{3} & \multicolumn{1}{c}{27} & \multicolumn{1}{c}{1008} & \multicolumn{1}{c}{120} & \multicolumn{3}{c}{$n=11$} & \\
\multicolumn{1}{c}{3} & \multicolumn{1}{c}{\textbf{24}} & \multicolumn{1}{c}{24} & \multicolumn{1}{c}{6} & \multicolumn{1}{c}{18} & \multicolumn{1}{c}{\textbf{7}} & \multicolumn{1}{c}{24} & \multicolumn{1}{c}{3} & \multicolumn{1}{c}{29} & \multicolumn{1}{c}{186} & \multicolumn{1}{c}{120} & \multicolumn{1}{c}{$\dmin$} & \multicolumn{1}{c}{LP} & \multicolumn{1}{c}{SB} \\
\multicolumn{1}{c}{4} & \multicolumn{1}{c}{12} & \multicolumn{1}{c}{6} & \multicolumn{1}{c}{6} & \multicolumn{1}{c}{19} & \multicolumn{1}{c}{\textbf{3}} & \multicolumn{1}{c}{6} & \multicolumn{1}{c}{2} & \multicolumn{1}{c}{30} & \multicolumn{1}{c}{\textbf{93}} & \multicolumn{1}{c}{120} & \multicolumn{1}{c}{18} & \multicolumn{1}{c}{39916800} & \multicolumn{1}{c}{3628800} \\
\multicolumn{1}{c}{5} & \multicolumn{1}{c}{\textbf{4}} & \multicolumn{1}{c}{6} & \multicolumn{1}{c}{2} & \multicolumn{1}{c}{20} & \multicolumn{1}{c}{\textbf{2}} & \multicolumn{1}{c}{6} & \multicolumn{1}{c}{2} & \multicolumn{1}{c}{31} & \multicolumn{1}{c}{\textbf{15}} & \multicolumn{1}{c}{24} & \multicolumn{1}{c}{19} & \multicolumn{1}{c}{\textbf{3326400}} & \multicolumn{1}{c}{3628800} \\
\multicolumn{1}{c}{6} & \multicolumn{1}{c}{\textbf{2}} & \multicolumn{1}{c}{2} & \multicolumn{1}{c}{2} & \multicolumn{1}{c}{21} & \multicolumn{1}{c}{\textbf{2}} & \multicolumn{1}{c}{2} & \multicolumn{1}{c}{2} & \multicolumn{1}{c}{32} & \multicolumn{1}{c}{\textbf{9}} & \multicolumn{1}{c}{24} & \multicolumn{1}{c}{31} & \multicolumn{1}{c}{359611} & \multicolumn{1}{c}{40320} \\
\multicolumn{1}{c}{} & \multicolumn{1}{c}{} & \multicolumn{1}{c}{} & \multicolumn{1}{c}{} & \multicolumn{1}{c}{} & \multicolumn{1}{c}{} & \multicolumn{1}{c}{} & \multicolumn{1}{c}{} & \multicolumn{1}{c}{33} & \multicolumn{1}{c}{\textbf{4}} & \multicolumn{1}{c}{24} & \multicolumn{1}{c}{33} & \multicolumn{1}{c}{193458} & \multicolumn{1}{c}{40320} \\
\multicolumn{4}{c}{$n =5$}    & \multicolumn{4}{c}{$n =8$}    & \multicolumn{1}{c}{34} & \multicolumn{1}{c}{\textbf{3}} & \multicolumn{1}{c}{6} & \multicolumn{1}{c}{34} & \multicolumn{1}{c}{177678} & \multicolumn{1}{c}{40320} \\
\multicolumn{1}{c}{$\dmin$} & \multicolumn{1}{c}{LP} & \multicolumn{1}{c}{SB} & \multicolumn{1}{c}{HB} & \multicolumn{1}{c}{$\dmin$} & \multicolumn{1}{c}{LP} & \multicolumn{1}{c}{SB} & \multicolumn{1}{c}{HB} & \multicolumn{1}{c}{35} & \multicolumn{1}{c}{\textbf{2}} & \multicolumn{1}{c}{6} & \multicolumn{1}{c}{35} & \multicolumn{1}{c}{94924} & \multicolumn{1}{c}{5040} \\
\multicolumn{1}{c}{6} & \multicolumn{1}{c}{120} & \multicolumn{1}{c}{24} & \multicolumn{1}{c}{8} & \multicolumn{1}{c}{12} & \multicolumn{1}{c}{40320} & \multicolumn{1}{c}{5040} & \multicolumn{1}{c}{64} & \multicolumn{1}{c}{36} & \multicolumn{1}{c}{\textbf{2}} & \multicolumn{1}{c}{2} & \multicolumn{1}{c}{37} & \multicolumn{1}{c}{66176} & \multicolumn{1}{c}{5040} \\
\multicolumn{1}{c}{7} & \multicolumn{1}{c}{\textbf{10}} & \multicolumn{1}{c}{24} & \multicolumn{1}{c}{4} & \multicolumn{1}{c}{13} & \multicolumn{1}{c}{\textbf{5040}} & \multicolumn{1}{c}{5040} & \multicolumn{1}{c}{32} &       &       &       & \multicolumn{1}{c}{41} & \multicolumn{1}{c}{33662} & \multicolumn{1}{c}{720} \\
\multicolumn{1}{c}{8} & \multicolumn{1}{c}{\textbf{5}} & \multicolumn{1}{c}{6} & \multicolumn{1}{c}{4} & \multicolumn{1}{c}{14} & \multicolumn{1}{c}{4135} & \multicolumn{1}{c}{720} & \multicolumn{1}{c}{32} & \multicolumn{3}{c}{$n=10$} & \multicolumn{1}{c}{42} & \multicolumn{1}{c}{26050} & \multicolumn{1}{c}{720} \\
\multicolumn{1}{c}{9} & \multicolumn{1}{c}{\textbf{2}} & \multicolumn{1}{c}{6} & \multicolumn{1}{c}{2} & \multicolumn{1}{c}{19} & \multicolumn{1}{c}{896} & \multicolumn{1}{c}{120} & \multicolumn{1}{c}{7} & \multicolumn{1}{c}{$\dmin$} & \multicolumn{1}{c}{LP} & \multicolumn{1}{c}{SB} & \multicolumn{1}{c}{43} & \multicolumn{1}{c}{11152} & \multicolumn{1}{c}{720} \\
\multicolumn{1}{c}{10} & \multicolumn{1}{c}{\textbf{2}} & \multicolumn{1}{c}{2} & \multicolumn{1}{c}{2} & \multicolumn{1}{c}{21} & \multicolumn{1}{c}{384} & \multicolumn{1}{c}{120} & \multicolumn{1}{c}{5} & \multicolumn{1}{c}{16} & \multicolumn{1}{c}{3628800} & \multicolumn{1}{c}{362880} & \multicolumn{1}{c}{44} & \multicolumn{1}{c}{8700} & \multicolumn{1}{c}{720} \\
      &       & \multicolumn{1}{c}{} & \multicolumn{1}{c}{} & \multicolumn{1}{c}{22} & \multicolumn{1}{c}{192} & \multicolumn{1}{c}{120} & \multicolumn{1}{c}{5} & \multicolumn{1}{c}{17} & \multicolumn{1}{c}{\textbf{329891}} & \multicolumn{1}{c}{362880} & \multicolumn{1}{c}{45} & \multicolumn{1}{c}{6349} & \multicolumn{1}{c}{720} \\
\multicolumn{4}{c}{$n =6$}    & \multicolumn{1}{c}{23} & \multicolumn{1}{c}{41} & \multicolumn{1}{c}{24} & \multicolumn{1}{c}{3} & \multicolumn{1}{c}{18} & \multicolumn{1}{c}{302400} & \multicolumn{1}{c}{40320} & \multicolumn{1}{c}{46} & \multicolumn{1}{c}{3541} & \multicolumn{1}{c}{120} \\
\multicolumn{1}{c}{$\dmin$} & \multicolumn{1}{c}{LP} & \multicolumn{1}{c}{SB} & \multicolumn{1}{c}{HB} & \multicolumn{1}{c}{24} & \multicolumn{1}{c}{\textbf{21}} & \multicolumn{1}{c}{24} & \multicolumn{1}{c}{3} & \multicolumn{1}{c}{27} & \multicolumn{1}{c}{49371} & \multicolumn{1}{c}{5040} & \multicolumn{1}{c}{47} & \multicolumn{1}{c}{222} & \multicolumn{1}{c}{120} \\
\multicolumn{1}{c}{8} & \multicolumn{1}{c}{720} & \multicolumn{1}{c}{120} & \multicolumn{1}{c}{14} & \multicolumn{1}{c}{25} & \multicolumn{1}{c}{\textbf{8}} & \multicolumn{1}{c}{24} & \multicolumn{1}{c}{2} & \multicolumn{1}{c}{29} & \multicolumn{1}{c}{21098} & \multicolumn{1}{c}{5040} & \multicolumn{1}{c}{48} & \multicolumn{1}{c}{\textbf{111}} & \multicolumn{1}{c}{120} \\
\multicolumn{1}{c}{9} & \multicolumn{1}{c}{\textbf{120}} & \multicolumn{1}{c}{120} & \multicolumn{1}{c}{7} & \multicolumn{1}{c}{26} & \multicolumn{1}{c}{\textbf{5}} & \multicolumn{1}{c}{6} & \multicolumn{1}{c}{2} & \multicolumn{1}{c}{31} & \multicolumn{1}{c}{9735} & \multicolumn{1}{c}{720} & \multicolumn{1}{c}{49} & \multicolumn{1}{c}{\textbf{17}} & \multicolumn{1}{c}{120} \\
\multicolumn{1}{c}{11} & \multicolumn{1}{c}{27} & \multicolumn{1}{c}{24} & \multicolumn{1}{c}{4} & \multicolumn{1}{c}{27} & \multicolumn{1}{c}{\textbf{2}} & \multicolumn{1}{c}{6} & \multicolumn{1}{c}{2} & \multicolumn{1}{c}{35} & \multicolumn{1}{c}{4995} & \multicolumn{1}{c}{720} & \multicolumn{1}{c}{50} & \multicolumn{1}{c}{\textbf{11}} & \multicolumn{1}{c}{24} \\
\multicolumn{1}{c}{12} & \multicolumn{1}{c}{\textbf{13}} & \multicolumn{1}{c}{24} & \multicolumn{1}{c}{4} & \multicolumn{1}{c}{28} & \multicolumn{1}{c}{\textbf{2}} & \multicolumn{1}{c}{2} & \multicolumn{1}{c}{2} & \multicolumn{1}{c}{36} & \multicolumn{1}{c}{3900} & \multicolumn{1}{c}{120} & \multicolumn{1}{c}{51} & \multicolumn{1}{c}{\textbf{5}} & \multicolumn{1}{c}{24} \\
\multicolumn{1}{c}{13} & \multicolumn{1}{c}{\textbf{6}} & \multicolumn{1}{c}{6} & \multicolumn{1}{c}{2} &       &       &       &       & \multicolumn{1}{c}{37} & \multicolumn{1}{c}{446} & \multicolumn{1}{c}{120} & \multicolumn{1}{c}{52} & \multicolumn{1}{c}{\textbf{4}} & \multicolumn{1}{c}{24} \\
\multicolumn{1}{c}{14} & \multicolumn{1}{c}{\textbf{4}} & \multicolumn{1}{c}{6} & \multicolumn{1}{c}{2} &       &       &       &       & \multicolumn{1}{c}{38} & \multicolumn{1}{c}{230} & \multicolumn{1}{c}{120} & \multicolumn{1}{c}{53} & \multicolumn{1}{c}{\textbf{3}} & \multicolumn{1}{c}{6} \\
\multicolumn{1}{c}{15} & \multicolumn{1}{c}{\textbf{2}} & \multicolumn{1}{c}{2} & \multicolumn{1}{c}{2} &       &       &       &       & \multicolumn{1}{c}{39} & \multicolumn{1}{c}{\textbf{55}} & \multicolumn{1}{c}{120} & \multicolumn{1}{c}{54} & \multicolumn{1}{c}{\textbf{2}} & \multicolumn{1}{c}{6} \\
      &       &       &       &       &       &       &       & \multicolumn{1}{c}{40} & \multicolumn{1}{c}{30} & \multicolumn{1}{c}{24} & \multicolumn{1}{c}{55} & \multicolumn{1}{c}{\textbf{2}} & \multicolumn{1}{c}{2} \\
      &       &       &       &       &       &       &       & \multicolumn{1}{c}{41} & \multicolumn{1}{c}{\textbf{11}} & \multicolumn{1}{c}{24} &       &       &  \\
\bottomrule
\end{tabular}%

	\label{tab:KT}
	\begin{tablenotes}
	%\item [${}^\syma$] Computed by solving the linear programming (LP) problem (\ref{eqn:LP_Dual}).
	%\item [${}^\symb$] Computed using the simplification of the LP problem discussed in Remark \ref{rem:1}.
	%\item [${}^\symc$] Computed by solving the dual problem (\ref{eqn:SDP_Dual}), using diagonalization techniques described in~\cite{ISDP}.
%	\item [${}^\symb$] Singleton bound (SB), published in~\cite{Barg2009}, equation (5).
%	\item [${}^\symc$] Hamming bound (HB), obtained using ball sizes calculated from known numbers of permutations with $k$ inversions,~see\cite{Margolius2001}.
%	\item []  {\em Note: the optimal objective values of (\ref{eqn:LP_Dual}) are not always integers. Above table created by taking numerical floor.}
	\item [${}^\syma,{}^\symb$] See footnotes on previous Table \ref{tab:SDP}.
	\end{tablenotes}
	\end{threeparttable}
	\vspace*{-15pt}
\end{table*}

\renewcommand{\Cen}{\mathcal{B}}

\renewcommand{\fn}{\footnote{Try to show, see~\cite{BI}, pp. 50-51., that $\Am{i}$ and $\Am{i}W$ are \emph{conjugacy-sums}, and $\Cen$ in (\ref{eqn:Cen}) is the \emph{center} of the adjacency algebra (\ref{eqn:AlgT}) for $\G=\Sym_n\times\Cyc$.}}

%We extend the ``diagonalization'' idea to a \emph{larger} subset of matrices.
A simple extension of the ``diagonalization'' idea to the following \emph{larger} subset of matrices, works reasonably well.
%Those familiar with linear algebra theory may notice that property iii) of Thm. \ref{thm:conjCC} implies ii), 
%as symmetric matrices that commute share \emph{common eigenspaces}. 
Property ii) of Theorem \ref{thm:conjCC} implies iii), as symmetric matrices that commute share \emph{common eigenspaces}. 
%Any two matrices in $\AlgO$ commute.
%To obtain a set of commutating matrices, 
As such,
we desire\fn~a subset $\Cen$ of $\AlgT$, with the property that any $M\in \Cen$, commutes with any $M' \in \AlgT$. 
%i.e. $\Cen\supseteq \AlgO$, whereby any two matrices in $\Cen$ commute.
%Consider the set $\Cen$ of matrices that satisfies 
Thus any two matrices in $\Cen$ commute. %We set
Such a subset $\Cen$ may be obtained  %\vspace*{-10pt}
\bea
\Cen = \left\{ \sum_{i=1}^\d (a_i  \Am{i}) + \sum_{i=1}^\d (a_{\d+i}  \Am{i} W)  : a \in \Real^{2\d} \right\}, \label{eqn:Cen}
\eea
where $W$ is an orthonormal, 0-1 matrix in $\Real^{\Sym_n\times\Sym_n}$, that satisfies $(W)_{x,y}=1$ if and only if $y \wo\Iv = x$ for any $x,y\in \Sym_n$. 
From (\ref{eqn:Cen}) we see $\Cen$ contains the set $\AlgO$ considered in Theorem \ref{thm:conjCC}.
\ifthenelse{\boolean{longver}}{
Also by the previous correspondence between $B_j$ and the orbital $\Orb_j$, one can check (see appendix) $W$ commutes with all of $\AlgT$ (and each $\Am{i}$). 
}
{
Also by the previous corresponding between $B_j$ and the orbital $\Orb_j$ (see~\cite{FH}), one can check $W$ commutes with all of $\AlgT$ (and each $\Am{i}$). 
%Because the longest element satisfies $\wo\Iv = \wo$, thus $W^T = W\Iv = W$. So $\Am{i} W$ are all symmetric, and $\Cen$ is a set of symmetric matrices. 
}
Because the longest element satisfies $\wo\Iv = \wo$, thus $W^T = W\Iv = W$. So $\Am{i} W$ are symmetric, and $\Cen$ is a set of symmetric matrices. 
%%CHECKING W commutes with all B_j excluded.
%For any adjacency matrix $B_j$ of the length CC, observe that $(W^T B_j W)_{x,y} = (B_j)_{x\wo\Iv,y\wo\Iv}$.
%Recall $(B_j)_{x,y}=1$ if and only if $(x,y)\in \Orb_j$, whereby $\Orb_j$ is an orbital the induced action of $\Sym_n\times \Cyc$ on $\Sym_n$. Hence $(B_j)_{x\wo\Iv,y\wo\Iv} = (B_j)_{x,y}$ because $(x,y)$ and $(x\wo\Iv,y\wo\Iv)$ both belong to same orbital. Hence for any $j$ we have $B_j W = W B_j$, which implies $\Cen$ is a set of commutating matrices.
%%END OF CHECKING
%The set $\Cen$ is not arbitrarily chosen, it equals the largest subset of $\AlgT$ such any $M\in \Cen$, commutes with any $M' \in \AlgT$. 
%That is loosely speaking, $\Cen$ is the \emph{``largest commutative part''} of $\AlgT$ and we exploit this commutativity, such that the positive semidefiniteness of any matrix in $\Cen$ can be checked via linear constraints (similarly as we showed before for $\AlgO$).

%\newcommand{\Bc}{\tilde{O}}
\newcommand{\Bc}[1]{\tilde{A}_{\Theta_n,#1}}

%Thus far our exposition of the intuition is complete. As described before for $A_i$, the matrices $A_i W$ can be similarly expressed using linear constraints. We are sparse with details here.
%We seek feasible $b\in \Real^\dsym$ that correspond to matrices in (\ref{eqn:Cen}). The matrices $A_i$ in (\ref{eqn:Cen}) have the expression $A_i = \sum_{i=1}^\d \p_{i,j} (U I_j U^T) $, see Thm. \ref{thm:conjCC}, which leads to expressing $\sum_{i=1}^\dsym a_i A_i$ using linear expressions, and the similar thing can also be done for $A_i W$. 
One technical lemma, that connects (\ref{eqn:Cen}) with the dual problem (\ref{eqn:SDP_Dual}), stands in way of finally describing our LP bound.
%The expression (\ref{eqn:Cen}) for $\Cen$ shows a clear connection with Thm. \ref{thm:conjCC} (that we exploit in the sequel). However we need one more lemma to be able to connect $\Cen$ 
This lemma involves a special subgroup $\Theta_n$ of $\Sym_n$, where $\Theta_n$ is also involved in a few final definitions. Let
$\Theta_n = \{\alpha\in\Sym_n: (\alpha,\alpha) \cdot \wo = \wo\}$, where $(\alpha,\alpha) \cdot \wo$ is computed using the action of $\Sym_n\times\Sym_n$ on $\Sym_n$. Let $\Bc{\ell}$ denote the symmetrized adjacency matrices belonging to the CC $(\Sym_n\times\Theta_n,\Sym_n)$, where there are $\dthe$ of them. Note $\dthe \leq \dsym$.

\begin{lem} \label{lem:1}
Let $\Am{i}$ and $\Bc{\ell}$ be the symmetrized adjacency matrices belonging to the conjugacy CC and $(\Sym_n\times\Theta_n,\Sym_n)$, respectively. 
%Let $W$ be a 0-1 orthonormal matrix defined as before.
Let $W$ be defined as before.
For $1\leq \ell\leq \dthe$ and $1 \leq i \leq 2\d$ there exists 0-1 coefficients $t_{\ell,i}$ that satisfy
\bea
  \!\!\!\!\!\Am{i} &=&  \sum_{\ell=1}^\dthe t_{\ell,i}   \Bc{\ell},~~
  \Am{i} W = \sum_{\ell=1}^\dthe t_{\ell,\d+i}   \Bc{\ell}. \label{eqn:T}
\eea
\end{lem}

\renewcommand{\Is}{{\tilde{\mathcal{I}}_{\Theta_n,\ell}}}
\renewcommand{\Gamma}{\gamma}

\ifthenelse{\boolean{longver}}{
See appendix for the proof of Lemma \ref{lem:1}. 
The coefficients $t_{\ell,i}$ satisfying (\ref{eqn:T}) are used to state the following main theorem. 
}{%The proof of Lemma \ref{lem:1} is excluded for space considerations.}
See~\cite{FH} for proof.
The coefficients $t_{\ell,i}$ satisfying (\ref{eqn:T}) are used to state the following main theorem.}
For $\Theta_n\subseteq \Sym_n$, let index subsets $\Is$ satisfy $\Bc{\ell} = \sum_{j \in \Is} \BmT{i}$. Using orbit-distances $\del_j$ w.r.t. $(\Sym_n\times\Cyc,\Sym_n)$ and the \Kendall metric $\dist(,)$, define constants $\Gamma_\ell$ that satisfy $\Gamma_\ell =  \max \{\del_j : j \in \Is\}$. 
%For $\ell\geq 2$, the sets $\Gamma_\ell$  partition the distance set.% corresponding to the Kendall-Tau distance metric.

\begin{thm}[LP Bound on $\lengthCC$ and $\D$] \label{thm:main}
%Let $\D$ be a subset of the distance set corresponding to $(\Sym_n\times\Cyc,\Sym_n)$ and the Kendall-Tau distance metric $\dist(,)$. 
Let $W$ be the 0-1 orthornormal matrix defined as before. %Let $I$ denote the identity matrix. 

%Let $\del_j$ be the orbit-distances w.r.t. $(\G,\Sym_n)$ and $\dist(,)$.

For $1\leq i,j\leq \d$, let constants $\p_{i,j}, c_j$ and matrices $U,I_j$ be obtained from Theorem \ref{thm:conjCC}. 
%Also obtain the matrices $U$ and $I_j$. 
Let matrices $M_{1,j}$ and $M_{2,j}$ satisfy $M_{1,j} = \frac{1}{2} (U I_j U^T) (I+W)$ and $M_{2,j} = \frac{1}{2}(U I_j U^T)  (I-W)$. 

For $1\leq \ell\leq \dthe$, let the constants $\Gamma_\ell$ be defined as above. For $1\leq i \leq 2\d$, let the coefficients $t_{\ell,i}$ satisfy (\ref{eqn:T}). 
%Let $\emptyset$ denote the empty set. 
Let $a^*$ in $\Real^{2\d}$ solve the following LP problem 
\begin{align} %\setlength{\jot}{-5pt}
      & ~~~~~~\min_{(a_1, a_2, \dots, a_{2\d}) \in \Real^{2\d}} \sum_{i=1}^{2\d} t_{1,i} \cdot a_i \label{eqn:LP_Dual} \\[-4pt]
\mbox{s.t. }
    %& ~~~~a_i \leq 0~~~\mbox{ for }~ 2 \le i \le \d \mbox{ with } \gamma_i \in \D, \nn
    &  \sum_{i=1}^{2\d} t_{\ell,i} \cdot a_i \leq 0~~~\mbox{ for }~ 2 \le \ell \le \dthe \mbox{ with } \Gamma_\ell \geq\D, \nonumber \\[-4pt]
         \sum_{i=1}^\d & (a_i +a_{\d+i} ) \cdot \p_{i,j}    \geq c_j \mbox{ for } 1 \leq j \leq \d  \mbox{ with }  M_{1,j} \neq 0 \nonumber \\[-4pt]
         \sum_{i=1}^\d & (a_i -a_{\d+i} ) \cdot \p_{i,j}   \geq c_j \mbox{ for } 1 \leq j \leq \d \mbox{ with } M_{2,j} \neq 0 \nonumber
\end{align}
Let $b_1^*$ and $\Asize(n,\D)$ respectively denote the optimal objective values of the dual problem (\ref{eqn:SDP_Dual}), and the optimal code size problem (\ref{eqn:opt1}), for $\G = \Sym_n\times\Cyc$ and the \Kendall metric $\dist(,)$ and $\D$. %and 
Then we have the following inequalities %for the Kendall-Tau metric
\[
\Asize(n,\D) \leq b_1^* \leq \sum_{i=1}^{2\d} t_{1,i} \cdot a_i^* .
\]
\end{thm}

As promised our main result Theorem \ref{thm:main} furnishes an LP bound on the optimal code size $\Asize(n,\D)$.
% more specifically on the optimal objective value $b_1^*$ of the dual problem (\ref{eqn:SDP_Dual}). 
\ifthenelse{\boolean{longver}}{
See appendix for proof.
}{See~\cite{FH} for proof.
}
The number $\dthe$ of matrices $\Bc{\ell}$ is given in the previous Table \ref{tab:dsym}, where observe $\dthe > \d$, but $\dthe$ is much reduced from $\dsym$. Table \ref{tab:KT} shows our computed LP bounds whereby $n$ is between $3$ and $11$.
Our proposed LP bound fails to completely answer the question posed (at the end) of Section \ref{sect:SDP}, but some initial success is obtained.
%We compare with the \emph{Singleton bound} (SB) in~\cite{Barg2009}, derived for finite lengths. 
Observe that our LP bound is at least as tight as the SB in the places highlighted in bold font. 
Improvements are mainly obtained when $\dmin$ is close to $n(n-1)/2$. 
Interestingly, these two bounds are useful for similar ranges of $\dmin$ (the SB is known to be non-trivial only when $\dmin\geq n$, see~\cite{Barg2009}).
For the case $n=3$ the LP and SDP bounds are equal,
%but for $n> 4$, the LP bound does not approximate well for all $\dmin$.
though unfortunately for $n >4$, our LP bound does worse than the HB, and the performance gap gets bigger for smaller $\dmin$. 
%While at this point we have little intuition to offer on how to seek out better feasible solutions to the dual program (\ref{eqn:SDP_Dual}), i
Inspired by~\cite{Barg2009} (which points out three regions with different asymptotics), it is tempting to conjecture that different strategies work for cases $\dmin < n$ and $\dmin \geq n$. 
The subset searched here works reasonably well for the latter case, 
for the former what are the ``good'' dual-feasible subsets?

\renewcommand{\dmin}{\del_{\scriptsize\mbox{min}}}
%For our LP bound, a computational bottleneck 
%%which currently prevents moving toward larger $n$, is not the solving of (\ref{eqn:LP_Dual}), but the 
%is the computation of ``maximum distances'' $\Gamma_\ell$, where $\Gamma_\ell =  \max \{\del_j : j \in \Is\}$. If one replaces $\Theta_n$ by $\Sym_n$, then expressions for $\Gamma_\ell$ are known, see~\cite{Geck2000}. 
%One concern is how large the number $\dthe$ of non-positive constraints in (\ref{eqn:LP_Dual}) can get. 
%While no rigorous analysis is given, note $\dthe \geq \size{\Sym_n}/\size{\Theta_n}$ and~\cite{Brink1999} states $\size{\Theta_n}$ grows at least \emph{exponential} in $n$.

\vspace*{-2pt}
One issue: no known efficient method to compute ``max. distances'' $\Gamma_\ell$, where $\Gamma_\ell =  \max \{\del_j : j \in \Is\}$. If one replaces $\Theta_n$ by $\Sym_n$ in the expression for $\Gamma_\ell$, (where $\Am{i} = \sum_{j\in \tilde{\Iss}_{\Sym_n,i}} \Bm{j}$),
then~\cite{Geck2000} has closed-forms for $\Gamma_\ell$. Also its is unclear how large the number $\size{\{1\leq \del_\ell \leq \dthe : \del_\ell\geq \dmin \}}$ of non-positive constraints could be. 
No rigorous analysis is done here, but see~\cite{Brink1999} for a characterization of $\Theta_n$. % has been characterized in. 

%Note that also replacing $\Theta_n$ by a larger group, one may reduce the number $\dthe$ of non-positive constraints in (\ref{eqn:LP_Dual}).

%\section{Conclusion}
%
%Using the symmetries of the Kendall-Tau metric, we considered SDP problems that deliver upper bounds on the optimal code sizes for various distance subsets $\D$. These SDP problems are found to have a large number of adjacency matrices, which motivates our efforts to obtain simpler techniques. Here we present a simple idea to consider a smaller set of SDP feasible solutions, which allows us to deal with LP problems with smaller number of optimization variables. We evaluated our proposed techniques, and found that while there is room for further improvement, we present some new upper bounds for this difficult problem.

%\vspace*{-2.5pt}
\section{Conclusion \& Future Directions}
%\vspace*{-2.5pt}
Motivated by recent work on solving SDP's with algebraic structure, we formulated the optimal code size problem w.r.t. \Kendall metric as a SDP, and propose using LP to search for solutions. The problem seems difficult, but we report modest improvement over a recent Singleton bound. 

%\vspace*{-2.5pt}
The interest is to progress toward (possibly) beating known Hamming bounds, for the cases $n \geq 6$ (other than those shown here).
We offer some future directions. As previously mentioned, it would be nice to analyze the subsets that should be searched (for $\dmin < n$). 
%to address the subsets to be searched for $\dmin < n$. 
Next, one might generalize to larger subsets where a manageable SDP (not a LP as here) is used for searching. 
Finally, one might seek a similar Fourier-type analysis as~\cite{Navon2007}, using \emph{representation-theoretic} techniques.

%--------------------------------- Bibliography ----------------------
%\newpagep
%\pagestyle{plain}

%%\noindent {\bf D. REFERENCES}
%\noindent {\bf REFERENCES

%\bibliographystyle{ieeetr}

%%\bibliography{t}
%%\bibliography{tzzt,fabian}
%\bibliography{fabian}


% http://www.ctan.org/tex-arc hive/biblio/bibtex/contrib/doc/

% http://www.ctan.org/tex-archive/macros/latex/contrib/supported/IEEEtran/bibtex
%\begin{spacing}{0.9}
%\bibliographystyle{IEEEtran}
%% argument is your BibTeX string definitions and bibliography database(s)
%\bibliography{IEEEabrv,fabian,LP}
%\end{spacing}

% ----- Document End -----------------------------------------------------

\ifthenelse{\boolean{longver}}{

}{
%\bibliographystyle{IEEEtran}
%%% argument is your BibTeX string definitions and bibliography database(s)
%\bibliography{fabian,LP}


}

\ifthenelse{\boolean{longver}}{

\newpage
\appendix

%\begin{proof}[Proof of Proposition \ref{thm:main}]

\subsubsection{SDP relaxation of optimal code size problem}
First we prove (\ref{eqn:SDP}) is a relaxation of (\ref{eqn:opt1}). In the following for a subset $\S $ of $ \Sym_n$, let $\S^2$ denote the product set $\S\times\S$. Let $\Real^{\Sym_n }$ denote the set of vectors with real number entries with index set $\Sym_n$.

\begin{proof}[Proof of Proposition \ref{pro:SDPrelax}]
Let $\mathcal{V}$ denote a solution of (\ref{eqn:opt1}),
 \textit{i.e.}, let $\size{\mathcal{V}}=\Asize(n,\D)$. 
  Identify the subset $\mathcal{V} $ of $ \Sym_n$
 with an 0-1 vector $\x $ in $ \Real^{\Sym_n}$, where $\x_g = 1$ if and only if $g \in \S$.
 We construct a matrix $M$ whose objective value in (\ref{eqn:SDP}) equals $\size{\S}$, \textit{i.e.}, $\Tr(J M) = \size{\S}$.
 Let $M = \frac{1 }{ \x^\tpr \x } \x\x^\tpr$,
\textit{i.e.}, $M = \frac{1}{\# \S} \x \x^\tpr$, and let $\1$ denote the all-ones vector.
 Observe that
 $\Tr(J M) = \Tr( \1 \1^\tpr  M )
 = \Tr( \1^\tpr M \1 )
 = \frac{1}{ \size{ \S } } \Tr( (\1^\tpr \x) (\x^\tpr \1) )
 = \frac{1}{ \size{ \S } } \Tr( \size{\S} \size{\S} )
 = \size{\S}$. 
 Next we show that matrix $M$ is a feasible solution to (\ref{eqn:SDP}).

Because $M = \frac{1}{\size{\S}} bb^T$, therefore $M$ is positive semidefinite and $\Tr(M) = 1$ is satisfied.
 %Let $\Asym_j$ be symmetrized orbital of the induced action of $\G$ on $\Sym_n \times \Sym_n$.
 Next observe 
 $\Tr(\Asym_j M) =  \size{ ( \Orbsym_j \cap \S^2 )} / \size{\S}$, so $\Tr(\Asym_j M)\geq 0$ is satisfied.
% Note $\del_j = \del(\Orbsym_j)=\dist(\g,\h)$ for any $(\g,\h)\in\Orbsym_j$.
% By our convention $\Orbsym_1= \{(g,g): g \in \Sym_n\}$ and $\dist(\Orbsym_1)=0$. 
 Now consider any $\g,\h \in \Sym_n$ where $\g\neq\h$.
 If $\g,\h \in \S$ then $\dist(\g,\h) \geq \D$. %, or equivalently $\del_j \in \D$
 By contraposition, if $\dist(\g,\h) < \D$ then $(x,y)\notin \S^2$.
 Let $(x,y)\in \Orbsym_j$ for some $j \geq 2$, then $\del_j < \D$ also implies $(x,y)\notin \S^2$, which in turn implies $\Tr(\Asym_j M) = 0$.
% This implies that for any $j \geq 2$ whereby $\del_j\geq 1$  , we must have $\Tr(\Asym_j M) = 0$.
\end{proof}
\vspace*{10pt}
\subsubsection{Matrix $W$ and set $\AlgT$}
Next we prove the orthonormal, 0-1 matrix $W$ commutes with all matrices in the set $\AlgT$. Recall $W$ is related to the longest element $\wo$, where for any $x,y\in \Sym_n$, we have $(W)_{x,y}=1$ if and only if $y \wo\Iv = x$. 

\renewcommand{\Bm}[1]{{A}_{\Cyc,#1}}
\begin{proof}
It suffices to show that $W$ commutes with any adjacency matrix $\Bm{j}$ of the length CC.
For any $\Bm{j}$, observe that $(W^T \Bm{j} W)_{x,y} = (\Bm{j})_{x\wo\Iv,y\wo\Iv}$.
Recall $(\Bm{j})_{x,y}=1$ if and only if $(x,y)\in \Orb_j$, whereby $\Orb_j$ is an orbital of the induced action of $\Sym_n\times \Cyc$ on $\Sym_n\times\Sym_n$. Hence $(\Bm{j})_{x\wo\Iv,y\wo\Iv} = (\Bm{j})_{x,y}$ because $(x,y)$ and $(x\wo\Iv,y\wo\Iv)$ both belong to same orbital. Hence for any $j$ we have $\Bm{j} W = W \Bm{j}$, which implies $W$ commutes with all of $\AlgT$. %$\Cen$ is a set of commutating matrices.
\end{proof}

\newcommand{\Z}{\mathcal{Z}}
\newcommand{\At}[1]{{A}_{\Z,#1}}
\newcommand{\ProdAct}[1]{\alpha#1 \beta\Iv}
\newcommand{\alp}{\alpha}
\renewcommand{\G}{\Z}
\renewcommand{\H}{\mathcal{H}}
\newcommand{\RightAct}[1]{\rho(#1)}

\vspace*{10pt}
\subsubsection{Technical Lemma \ref{lem:1}}
To show Lemma \ref{lem:1} we need to first establish a relationship between the adjacency matrices $\At{i}$ of the CC $(\Sym_n\times\Z,\Sym_n)$, where $\Z$ is a subgroup of $\Sym_n$, with orbits on $\Sym_n$, of a subgroup of $\Sym_n\times\Sym_n$ that is related to $\Z$. Recall our definition of the action of any $\tup{\alp,\beta}$ in $\Sym_n \times \Sym_n$ on any $x \in \Sym_n$, given as $\tup{\alp,\beta}x = \ProdAct{x}$. 
For any subgroup $\Z$ of $\Sym_n$, denote the subgroup $ \{\tup{\beta,\beta}: \beta \in \G \}$ of $\Sym_n \times \Sym_n$ by $\H_\G$. 
Let $\C_{\G,1},\C_{\G,1},\cdots,\C_{\G,r}$ denote the $r$ orbits, obtained from the action of $\H_\G$ on $\Sym_n$. 
Each orbit $\C_{\G,i}$ is called a \textbf{conjugacy class} (of the action of $\H_\G$ on $\Sym_n$). 
%Recall the group $\Cyc$ that satisfies $\Cyc = \{e,\wo\}$, whereby $e$ and $\wo$ denote the identity, and longest element, respectively.
Let $\rho(\beta)$ denote the 0-1 matrix where $(\rho(\beta))_{x,y}=1$ if and only if $y \beta\Iv = x$ (\emph{i.e.}, by our previous definition, $W=\rho(\wo)$).
We claim a one-to-one correspondence between some conjugacy class $\C_{\G,i}$ and some adjacency matrix $\At{i}$ of the CC $(\Sym_n\times\Z,\Sym_n)$, given as
\bea
	\At{i} = \sum_{\beta \in \C_{\G,i}} \RightAct{\beta}.\label{eqn:adj1}
\eea
By this claim the number $r$ of conjugacy classes $\C_{\G,i}$, equals the number $d$ of adjacency matrices $\At{i}$. To show (\ref{eqn:adj1}), consider the following.

First, we establish the one-to-one correspondence.
By the definition of the orbital $\Orb_i$, for any $\tup {x,y}, \tup{\tilde{x},\tilde{y}} \in \Orb_i$ there exists some $\alpha\in\Sym_n$ and $\beta\in \G$ such that $\tilde{x} = \ProdAct{x}$ and $\tilde{y} = \ProdAct{y}$. Equivalently for any $\tup {x,y}, \tup{\tilde{x},\tilde{y}} \in \Orb_i$, there exists some $\beta \in \G$ that satisfies $\tilde{x}\Iv \tilde{y} = \beta x\Iv y \beta\Iv $, which means that $\tilde{x}\Iv \tilde{y}$ and $x\Iv y $ are both in $\C_{\G,i}$. 
%This establishes the said one-to-one correspondence.
Note that $\sum_{\beta \in \C_{\G,i}} \RightAct{\beta}$ is a 0-1 matrix, and by the definition of $\rho(\beta)$ we conclude
\bea
\left(\sum_{\beta \in \C_{\G,i}} \RightAct{\beta}\right)_{x,y} = 1,~~~ \mbox{ if and only if }~~~
x\Iv y \in \C_{\G,i},  \nonumber %~~~ \mbox{ if and only if }~~~
%\tup{x,y} \in \Orb_i.  \nonumber %\label{eqn:proadj}
\eea
if and only if $\tup{x,y} \in \Orb_i$. 
This establishes (\ref{eqn:adj1}) by 
by referring to the original definition of $\At{i}$ from $\Orb_i$. % we conclude $\At{i} = \sum_{\beta \in \C_{\G,i}} \RightAct{\beta}$.

\newcommand{\woC}{\mathcal{P}_{\Sym_n,i}}
\renewcommand{\Iss}{{\mathcal{I}}_{\Theta_n,\ell}}
\newcommand{\dlen}{{d_{\Cyc}}}
\renewcommand{\dthe}{{d_{\Theta_n}}}
\newcommand{\Js}{{\mathbb{J}}_i}
\begin{proof}[Proof of Lemma \ref{lem:1}] \newcommand{\W}{W}
Denote a set $ \{ \beta \wo : \beta \in \C_{\Sym_n,i}\}$ of elements in $\Sym_n$ by $\woC$.
%$\woC$ to be a set containing elements of $\Sym_n$, that satisfies $\woC = \{ \beta \wo : \beta \in \C_{\Sym_n,i}\}$. 
Hence $\woC$ is obtained using the conjugacy class $\C_{\Sym_n,i}$ and the longest element $\wo$.
Let $\Am{i}$ be an adjacency matrix of the conjugacy CC, and $W = \rho(\wo)$. 
It follows $\Am{i}\W = \sum_{\beta \in \C_{\Sym_n,i}} \rho(\beta \wo) = \sum_{\beta \in \woC} \rho(\beta)$. % because $\wo\Iv = \wo$. 
%Recall $\Theta_n$ is a subgroup of $\Sym_n$, related to the longest element $\wo$, that satisfies $\Theta_n = \{\alpha\in\Sym_n: (\alpha,\alpha) \cdot \wo = \wo\}$. 
Because $\Cyc $ is a subgroup of $ \Theta_n$, so for each conjugacy class $\C_{\Theta_n,\ell}$ there exists index sets $\Iss$ satisfying $\C_{\Theta_n,\ell} = \cup_{j\in\Iss} \C_{\Cyc,j}$. 
The sets $\Iss$ partition $\{1,2,\cdots, \dlen\}$. 
%Let $\dthe$ denote the number of conjugacy classes $\C_{\Theta_n,\ell}$.
For $1\leq i \leq \d$, we claim there exists \emph{new} index sets $\mathbb{I}_i$ and $\Js$ that satisfy 
\bea
	\C_{\Sym_n,i} &=& \cup_{\ell \in \mathbb{I}_i} \C_{\Theta_n,\ell}, \label{eqn:woC1} \\
	\woC &=& \cup_{\ell \in \Js } \C_{\Theta_n,\ell}. \label{eqn:woC}
\eea
If the claim holds, Lemma \ref{lem:1} is easily proved as follows. 

\newcommand{\Jss}{{\tilde{\mathbb{J}}}_{i}}
\renewcommand{\Bc}[1]{{A}_{\Theta_n,#1}}
\newcommand{\BcT}[1]{\tilde{A}_{\Theta_n,#1}}
\renewcommand{\dthe}{{\tilde{d}_{\Theta_n}}}
By the previously established (\ref{eqn:adj1}), we can write $\Bc{\ell} = \sum_{\beta\in\C_{\Theta_n,\ell}} \rho(\beta)$, where $\Bc{\ell}$ is an adjacency matrix of the CC $(\Sym_n\times\Theta_n,\Sym_n)$. By (\ref{eqn:adj1}) again, an adjacency matrix $\Am{i}$ of the conjugacy CC satisfies $\Am{i} = \sum_{\beta\in\C_{\Sym_n,i}} \rho(\beta)$, so (\ref{eqn:woC1}) implies $\Am{i} = \sum_{\ell\in \mathbb{I}_i} \Bc{\ell}$. 
Also because $\Am{i}\W = \sum_{\beta \in \C_{\Sym_n,i}} \rho(\beta \wo)$, by definition of $\woC$ then (\ref{eqn:woC}) implies $\Am{i} \W= \sum_{\ell\in \Js} \Bc{\ell}$. But because both $\Am{i}$ and $\Am{i}\W$ are symmetric 0-1 matrices, there must exist sets ${\tilde{\mathbb{I}}}_{i}$ and $\Jss$ to express $\Am{i}$ and $\Am{i}\W$ in terms of symmetrized adjacency matrices $\BcT{\ell}$, i.e.
\bea
  \Am{i} &=& \sum_{\ell \in {\tilde{\mathbb{I}}}_{i}} \BcT{\ell} = \sum_{\ell=1}^\dthe t_{\ell,i}   \cdot \BcT{\ell}, \nn
  \Am{i}\W &=& \sum_{\ell \in \Jss} \BcT{\ell} = \sum_{\ell=1}^\dthe t_{\ell,\d+i}  \cdot \BcT{\ell},  \nonumber
\eea
where $t_{\ell,i}$ are coefficients appearing in the lemma statement.

We end by showing the previous claims.
The first identity (\ref{eqn:woC1}) follows easily from the fact $\Theta_n \subseteq \Sym_n$. % the sets $\Is$ partition $\{1,2,\cdots, \dthe\}$, where $\dthe$ equals the number of conjugacy classes $\C_{\Theta_n,\ell}$. 
The second identity (\ref{eqn:woC}) follows by arguing if $\C_{\Theta_n,\ell} \cap \woC \neq \emptyset$ then $\C_{\Theta_n,\ell} \subset \woC$. Consider some $ x \wo\in \C_{\Theta_n,\ell} \cap \woC$, where $x \in \C_{\Sym_n,i}$. By definition of the conjugacy class $\C_{\Theta_n,\ell} = \{(\alpha,\alpha)\cdot  x\wo : \alpha \in \Theta_n\}$. By definition of the group $\Theta_n$ we have $(\alpha,\alpha)\cdot  x \wo= \alpha  x \wo\alpha\Iv = (\alpha x \alpha\Iv) (\alpha \wo \alpha\Iv) = (\alpha x \alpha\Iv)\wo$, and it follows $(\alpha x \alpha\Iv)\wo$ is also in $ \woC$ as $\alpha x \alpha\Iv \in \C_{\Sym_n,i}$. 
%Hence the argument we sought is complete, and the sets $\Js$ also partition $\{1,2,\cdots, \dthe\}$.
Hence (\ref{eqn:woC}) is shown.
\end{proof}

\newcommand{\z}{z}
\newcommand{\W}{W}
\renewcommand{\dthe}{{\tilde{d}_{\Theta_n}}}
\renewcommand{\dsym}{{\tilde{d}_{\Cyc}}}
\renewcommand{\Bm}[1]{\tilde{A}_{\Cyc,#1}}
\vspace*{10pt}
\subsubsection{Main theorem}
%Finally, the following is the proof of our main Theorem \ref{thm:main}.
Finally, the following verifies that our proposed LP bound (\ref{eqn:LP_Dual}) indeed provides an upper bound to the optimal value of the dual problem (\ref{eqn:SDP_Dual}).
\begin{proof}[Proof of Theorem \ref{thm:main}]
For some $b\in \Real^\dsym$, $\z\in \Real^\dthe$ and $a\in \Real^{2\d}$, we have the following chain of equalities 
\bea
\!\!\!\!\!\!\!\!\sum_{j=1}^\dsym  b_j \cdot \Bm{j} \!\!&\stackrel{(a)}{=}& \!\!\sum_{\ell=1}^\dthe \z_\ell \cdot \Bc{\ell} \nn
&\stackrel{(b)}{=} & \!\!\sum_{i=1}^d a_i \cdot \Am{i} + \sum_{i=1}^\d a_{\d+i} \cdot \Am{i} \W, 
\label{eqn:BAW2}
\eea
where $(a)$ follows by setting $b_j=z_\ell$ if $j\in \Is$, and $(b)$ follows by setting $\z_\ell = \sum_{i=1}^{2\d}t_{\ell,i} \cdot a_i$ for coefficients $t_{\ell,i}$ that satisfy (\ref{eqn:T}).
The theorem will be proved by showing for any feasible $a$ in $\Real^{2\d}$ to (\ref{eqn:LP_Dual}), there corresponds some  some feasible $b$ in $\Real^{\dsym}$ to (\ref{eqn:SDP_Dual}) by the above relationship (\ref{eqn:BAW2}).

%First we show that the objectives of (\ref{eqn:SDP_Dual}) and (\ref{eqn:LP_Dual}) are equal. Because $I=\tilde{B}_1=\Bc_1$ (or equivalently $\tilde{I}_{\Theta_n,1}=\{1\}$) we have $b_1 = z_1  = \sum_{i=1}^{2\d}t_{1,i} \cdot a_i$.  

Firstly the objectives of (\ref{eqn:SDP_Dual}) and (\ref{eqn:LP_Dual}) are equal because 
%$I=\tilde{B}_1=\Bc_1$ (or equivalently $\tilde{I}_{\Theta_n,1}=\{1\}$) we have 
$b_1 = z_1  = \sum_{i=1}^{2\d}t_{1,i} \cdot a_i$.
Let $a$ satisfy the second constraint of (\ref{eqn:LP_Dual}) and let $b$ satisfy (\ref{eqn:BAW2}). By $\z_\ell = \sum_{i=1}^{2\d}t_{\ell,i} \cdot a_i$, if $\Gamma_\ell \geq \D $ we have $\z_\ell \leq 0$. Because $\Gamma_\ell= \max \{\del_j : j \in \Is\}$, then for any $j \in \Is$ such that $\del_j  \geq \D$, we must have $b_j \leq 0$. Finally $\cup_{\ell=2}^{\dthe} \{\del_j : j \in \Is\} = \{\del_2,\del_3,\cdots,\del_\dsym \}$, implying that $b_j \leq 0$ for all $j\geq 2$ whereby $\del_j \geq \D$, therefore $b$ satisfies the non-positive constraint of (\ref{eqn:SDP_Dual}).

Next consider the matrices $M_{1,j}$ and $M_{2,j}$ given in the theorem statement. Note $M_{1,j}+M_{2,j} = U I_j U^T$ and $M_{1,j}-M_{2,j} = (U I_j U^T) W$.
Using $\Am{i} = \sum_{j=1}^d p_{i,j}\cdot (U I_j U^T)$ in Theorem \ref{thm:conjCC} we express
\bea
\Am{i} &=& \sum_{j=1}^\d p_{i,j}\cdot M_{1,j} + \sum_{j=1}^\d p_{i,j}\cdot M_{2,j}, \nn
\Am{i}\W &=& \sum_{j=1}^\d p_{i,j}\cdot M_{1,j} - \sum_{j=1}^\d p_{i,j}\cdot M_{2,j}, \nn
J &=& \sum_{j=1}^\d c_j \cdot M_{1,j} + \sum_{j=1}^\d c_j \cdot M_{2,j}, \label{eqn:thmmain}
\eea
where we claim (shown below) that the matrices $M_{1,j}$ and $M_{2,j}$ are i) all symmetric, and ii) have eigenvalues only $0$ or $1$, and iii) $M_{1,j}M_{2,j}=0$ and iv) $\sum_{j=1}^d (M_{1,j} +M_{2,j}) =\sum_{j=1}^d (U I_j U^T) = I$. For any $a$ satisfying the last two constraints of (\ref{eqn:LP_Dual}), then (\ref{eqn:thmmain}) implies $\sum_{i=1}^d a_i \cdot \Am{i} + \sum_{i=1}^\d a_{\d+i} \cdot \Am{i} \W - J$ is positive semidefinite. Then for $b$ that corresponds by (\ref{eqn:BAW2}) to such an $a$, we will have $\sum_{j=1}^\dsym b_j \cdot \Bm{j}-J$ satisfying the positive semidefinite constraint in (\ref{eqn:SDP_Dual}).

% and (\ref{eqn:BAW2}) implies $\sum_{j=1}^\dsym b_j \cdot \Bm{j}$ is also positive semidefinite. Hence the corresponding $b$ to $a$ from (\ref{eqn:BAW2}), statisfies the p

To finish the proof we address the above claims i) and ii), whereby iii) and iv) will then follow from similar arguments. 
Claim i) follows because all matrices $U I_j U^T$ commute with all matrices $\Am{i}$, see Theorem \ref{thm:conjCC}.
%lie in $\AlgO$, see~\cite{Tarnanen}, and $\AlgO\subseteq\Cen$ where the set $\Cen$ is given in (\ref{eqn:Cen}). 
Recall $W$ commutes with all matrices in $\Cen$, therefore $W$ commutes with all matrices $\Am{i}$, which implies $W$ commute with all $U I_j U^T$. This implies $M_{1,j}$ and $M_{2,j}$ are symmetric, since both $W$ and $U I_j U^T$ are symmetric.

Claim ii) follows because $\wo\Iv = \wo$, and it can be verified that $W^T = W\Iv =W$, which implies that the possible eigenvalues of $W$ are $-1$ and $1$. Thus the possible eigenvalues of matrices $M_{1,j}$ and $M_{2,j}$ are $0$ or $1$.
%We claim that the product of any two matrices in $\{ M_{i,j}: i=1 \mbox{ or } 2, \mbox{ and } 1\leq j \leq \d\} $ . 
%Let $a$ satisfy the last two constraints of (\ref{eqn:LP_Dual}) and let $b$ satisfy (\ref{eqn:BAW2}). Hence (\ref{eqn:thmmain}) implies that 
\end{proof}
}{
%False do nothing
}

\end{document}